\documentclass[journal]{IEEEtran}
\usepackage{amsmath}
\usepackage{amsthm}
\usepackage{epsfig}
\usepackage{graphicx}
\usepackage{grffile}
\usepackage{times}
\usepackage{url}
\usepackage{color}
\usepackage{cite}
\usepackage{marvosym}
\usepackage{algorithmicx}
\usepackage[ruled]{algorithm}
\usepackage{algpseudocode}
\usepackage{multirow}
\usepackage{stfloats}
\usepackage{amssymb}
\usepackage{slashbox} 
\usepackage{algorithmicx}
\usepackage[ruled]{algorithm}
\usepackage{algpseudocode}
\usepackage{epstopdf}

\newtheorem{theorem}{Theorem}
\newtheorem{lemma}{Lemma}

\begin{document}
\newtheorem{notheorem}{Theorem}

\title{Thwarting Selfish Behavior in 802.11 WLANs}

\author{Albert Banchs, Jorge Ortin, Andres Garcia-Saavedra, Douglas J.~Leith and Pablo Serrano\thanks{A.~Banchs is with the University Carlos III of Madrid and with the Institute IMDEA Networks, Spain. J.~Ortin, A.~Garcia-Saavedra and P.~Serrano are with the University Carlos III of Madrid, Spain. D.~J.~Leith is with Hamilton Institute, Ireland.}} 


\maketitle

\begin{abstract} 
The 802.11e standard enables user configuration of several MAC parameters, making WLANs vulnerable to users that selfishly configure these parameters to gain throughput. In this paper we propose a novel distributed algorithm to thwart such selfish behavior. The key idea of the algorithm is for honest stations to react, upon detecting a selfish station, by using a more aggressive configuration that penalizes this station. We show that the proposed algorithm guarantees global stability while providing good response times. By conducting a game theoretic analysis of the algorithm based on \emph{repeated games}, we also show its effectiveness against selfish stations. Simulation results confirm that the proposed algorithm optimizes throughput performance while discouraging selfish behavior. We also present an experimental prototype of the proposed algorithm demonstrating that it can be implementated on commodity hardware.
\end{abstract}

\graphicspath{{}}

\section{Introduction}\label{sec-intro}

The mechanisms defined in 802.11e, which have been incorporated into the revised version of the 802.11 standard, rely on a number of configurable parameters that can be modified by a simple command. This gives users total control of the contention parameters used by their wireless adapter and allows them to modify the behavior of the wireless interface. In this framework, \emph{selfish} behavior is particularly tempting: users can very easily configure the 802.11e parameters of their station with aggressive values that increase their share of the medium at the expense of the other users. Such selfish behavior can lead to severe unfairness in throughput distribution.

A number of works in the literature have addressed the above selfishness problem. The approaches proposed can be classified in centralized \cite{toledo,domino,serrano,ilenia} and distributed \cite{cagalj05,konorski06,buttyan08}. The advantage of distributed approaches is that they do not rely on a central authority and thus can be used both in infrastructure and ad-hoc modes, in contrast to centralized approaches which can only be used in infrastructure mode. In this paper, we propose a novel \emph{distributed} approach based on game theory to address the selfishness problem.

Game theory is a discipline aimed at modeling situations in which decision-makers or players have to choose specific actions and obtain a gain that depends on the actions taken by all the players in the game. In our problem, the players are the 802.11e stations striving to obtain as much throughput as possible from the WLAN. Previous game theoretic analyses of WLAN \cite{buttyan08} have shown that, if selfish stations are not penalized, the WLAN naturally tends to either great unfairness or network collapse. Following this result, in this paper we focus on the design of a penalizing mechanism in which any player who misbehaves will be punished by other players and thus will have no incentive to misbehave.

A key challenge when designing such a penalizing scheme is to carefully adjust the punishment inflicted to a misbehaving station. Indeed, if the punishment is not severe, a selfish station could benefit from misbehaving. However, an overreaction could trigger the punishment of other stations leading to an endless loop of punishments. Our design makes use of Lyapunov stability theory to address this challenge. In particular, one of the key novelties of the approach proposed in this paper is the combination of Lyapunov stability theory techniques (which guarantee the stability and convergence of the algorithm) and game theory techniques (which guarantee protection against selfish behavior). The main contributions of our paper are as follows:
\begin{itemize}
\item We propose a novel distributed algorithm that penalizes selfish stations by making use of a more aggressive configuration of the 802.11e parameters upon detecting a misbehaving station.
\item We conduct a stability analysis of the algorithm to show that when all stations implement our algorithm, the WLAN converges to the optimal point of operation.
\item We conduct a game theoretic analysis based on repeated games that shows that a station cannot benefit by deviating from the algorithm.
\item We extensively evaluate the performance of the proposed algorithm via simulation under a wide variety of conditions that confirm its good properties.
\item We show the feasibility of implementing the algorithm by deploying a prototype and evaluating it in a small experimental testbed. 
\end{itemize}

The rest of the paper is structured as follows. In Section~\ref{sec-selfishness} we expose the selfishness problem in 802.11 and model it from a game theoretic standpoint. Section~\ref{sec-algorithm} presents the algorithm proposed. The algorithm is evaluated analytically in Section \ref{sec-analysis}: we first analyze its performance when all stations implement the algorithm and then conduct a game theoretic analysis for the case when stations may deviate from the algorithm. The performance of the algorithm is exhaustively evaluated via simulation in Section~\ref{sec-performance} and its feasibility of implementation is validated in Section~\ref{sec-experiment} by means of a prototype. Finally, Section~\ref{sec-conclusions} closes the paper with some concluding remarks.

\section{Selfishness in 802.11}\label{sec-selfishness}

In this section, we briefly summarize the EDCA mechanism of 802.11e and identify the selfishness problem. Then, we present a game theoretic model of this problem.

\subsection{802.11e EDCA}\label{sec-edca}

The 802.11e EDCA mechanism works as follows. When a station has a new frame to transmit, it senses the channel. If the channel remains idle for a period of time equal to the $AIFS$ parameter, the station transmits. Otherwise, if the channel is detected busy, the station monitors the channel until it is measured idle for an $AIFS$ time and then executes a backoff process.

When the backoff process starts, the station computes a random number uniformly distributed in the range $(0,CW-1)$, and initializes its backoff time counter with this value. $CW$ is called the contention window and for the first transmission attempt the minimum value is used (i.e., $CW=CW_{min}$) . In case of a collision $CW$ is doubled, up to a maximum value $CW_{max}$.

As long as the channel is sensed idle, the backoff time counter is decremented once every time slot $T_e$, and ``frozen'' when a transmission is detected on the channel. When the backoff time counter reaches zero, the station accesses the channel in the next time slot. Upon accessing the channel, a station can transmit several consecutive frames for a duration given by the $TXOP$ parameter.

A collision occurs when two or more stations start transmitting simultaneously. An acknowledgment (Ack) frame is used to notify the transmitting station that the frame has been successfully received. If the Ack is not received within a given timeout, the station reschedules the transmission by reentering the backoff process. If the number of failed attempts reaches a predetermined retry limit, the frame is discarded. Once the backoff process is completed, $CW$ is set again to $CW_{min}$. 

As it can be seen from the above description, the behavior of a station depends on a number of parameters, namely $CW_{min}$, $CW_{max}$, $AIFS$ and $TXOP$. As these are (according to the standard) configurable parameters whose setting can be modified by means of simple commands, a user can easily configure these parameters selfishly to gain extra throughput. We refer to this as the \emph{problem of selfishness in 802.11}.

\subsection{Game theoretic model}\label{sec-model}

The above problem of selfishness in 802.11 can be modeled using game theory. Indeed, game theory is a discipline aimed at modeling situations like the above in which players have to choose specific actions that have mutual, possibly conflicting, consequences. In our case, the players are the wireless stations which configure their 802.11e parameters to obtain as much throughput as possible.

The simplest way to model the interaction between players is by means of a static game. In a static game, each player makes a single move and all moves are made simultaneously. In our problem, this means that each station chooses its configuration at the beginning of the game, without any knowledge of the configuration chosen by the other stations, and keeps this configuration for the entire duration of the game.

The modeling of the selfishness problem in 802.11 based on static games \cite{cagalj05,konorski06} leads to the following two families of Nash equilibria: in the first family, there is one player that receives a non-null throughput while the rest of the players receive a null throughput, and in the second family, all players receive a null throughput (the latter is known as the \emph{tragedy of the commons}: the selfish behavior of each player leads to a tremendous misuse of the public good).
 
Both of the above families of solutions are highly undesirable, as they lead either to extreme unfairness or network collapse. One alternative to avoid these undesirable solutions is to allow that users make new moves (i.e., change their 802.11e configuration) during the game. This can be modeled by making use of the theory of \emph{repeated games}~\cite{repeated}. With repeated games, time is divided in stages and a player can take new decisions at each stage based on the observed behavior of the other players in the previous stages. Therefore, if a selfish station is detected to misbehave, the other stations can \emph{punish} this station and thus discourage such behavior.

\subsection{Related work}

Making use of repeated games, \cite{cagalj05} and \cite{konorski06} have proposed two approaches to address the problem of selfishness in 802.11. The approach proposed by \cite{cagalj05} is based on \emph{selective jamming}: if a station detects that another station is misbehaving, thereafter it listens to its transmitted packets and switches to transmission mode, \emph{jamming} enough bits so that the packets cannot be properly recovered at the receiver. While the use of jamming punishes misbehaving stations, it has the main drawback of relying on functionality not available in current wireless devices. Indeed, the implementation of such a jamming mechanism would need to be performed at the hardware level and entails substantial complexity.

The approach proposed by \cite{konorski06} does not suffer from the above drawback but addresses only two types of misbehaving stations: the so-called selfish stations, with $CW = 2$, and the so-called greedy stations, with $CW = 1$. While the scheme proposed is effective when dealing with these two particular configurations, other $CW$ configurations that may greatly benefit selfish stations are neither detected nor punished by this mechanism, as we show in the simulation results of Section~\ref{sec-comparison}. Additionally, the algorithm of \cite{konorski06} is based on heuristics that do not guarantee quick convergence, and indeed we show in a further simulation result in Section~\ref{sec-perturbations} that this approach may suffer from convergence issues.

In this paper, we propose a novel approach based on repeated games that, in contrast to the previous two approaches, relies exclusively on functionality readily available in current wireless devices and is effective against any selfish configuration. Additionally, by relying on Lyapunov stability techniques, our approach is guaranteed to quickly converge to the desired point of operation.

In addition to \cite{cagalj05,konorski06}, a number of additional works have been devoted to address selfishness in wireless networks from a game theoretic point of view \cite{MacKenzieW03,Altman02,kesidis,Inaltekin}. Besides focusing on a different MAC protocol, these studies differ from ours in that they consider some kind of transmission cost or pricing mechanism that plays a key role in the resulting equilibria. In contrast to these approaches, we achieve the desired equilibrium by means of a penalizing mechanism only.

The works in \cite{toledo,serrano,domino,ilenia} address, like ours, the issue of selfish stations in 802.11 WLANs. However, in contrast to our distributed algorithm, these works propose a centralized approach and therefore can only be applied to a WLAN operating in infrastructure mode. Additionally, many of these approaches only address the detection of misbehaving stations while our approach not only detects but also punishes selfish stations.

Substantial work in the literature has also focused on the design of stable adaptive algorithms \cite{tmc2,monet,hollow,boggia,tmc1,doc}. A major difference between our algorithm and these approaches is that they build on local stability analysis while we rely on Lyapunov stability theory, which ensures global asymptotic stability and hence provides stronger guarantees. Indeed, with \cite{tmc2,monet,hollow,boggia,tmc1,doc} convergence is only guaranteed as long as the initial point is sufficiently close to the stable point of operation, while we guarantee convergence for any initial point of operation.

Perhaps the most closely related to this paper is our previous work of \cite{doc}, which uses a similar technique to counteract selfish stations, based also on repeated games. However, both the scope of the work and the algorithm design are substantially different. Indeed, while \cite{doc} focuses on distributed opportunistic scheduling, here we address the problem of selfishness in 802.11. Furthermore, \cite{doc} relies on local linearized analysis, while here we use Lyapunov theory for the global design and analysis of the algorithm. As a consequence, the algorithm proposed in this paper provides much stronger guarantees on stability and convergence than that of \cite{doc}.


\section{GAS Algorithm}\label{sec-algorithm}

In this section, we present our algorithm to address the problem of selfishness in 802.11, which we call \emph{Game-theoretic Adaptive Stable} (GAS) algorithm. In the following, we first present the objectives pursued and then describe the algorithm design to achieve these objectives.

\subsection{Algorithm objectives and scope}\label{sec-objectives}

The central objective of the GAS algorithm is to drive the configuration of the 802.11e EDCA parameters to the target values that maximize the overall WLAN performance. To achieve this objective, GAS enforces that a selfish station cannot benefit from using a different configuration, which provides stations with an incentive to use the target configuration.

Following the arguments given in \cite{comnet,tvt,cagalj05}, in this paper we aim at the following setting of the four EDCA parameters, 
which maximizes the throughput performance of the WLAN (hereafter we refer to this setting as the \emph{target configuration} or \emph{optimal configuration}):
\begin{itemize}
\item The $AIFS$ parameter is set to its minimum value ($AIFS = DIFS$).
\item The $TXOP$ parameter is set such that one packet is transmitted upon accessing the channel ($TXOP = 1$~packet).
\item The maximum backoff stage $m$ is set equal to 0.\footnote{The maximum backoff stage is defined as the number of times that the $CW$ is doubled until reaching $CW_{max}$ (i.e., $CW_{max} = 2^mCW_{min}$).} This yields the same value for $CW_{min}$ and $CW_{max}$ (i.e., $CW_{max} = CW_{min}$); in the following, we refer to this value simply as $CW$.
\item The $CW$ parameter is set equal to the value that, given the above setting for the other parameters, maximizes the throughput of the WLAN when all stations are saturated. Hereafter, we refer to this value as $CW_{opt}$.
\end{itemize}

With the above, the objective of the GAS algorithm can be reformulated as to achieve the following two goals: ($i$) when all stations implement GAS (i.e., they are well-behaved), the system should converge to the target configuration given above; and ($ii$) if a selfish station misbehaves (by using a different configuration from the target one), this station should not obtain any benefit from such misbehavior.

In the following, we address the design of the GAS algorithm. 
Like the previous works of \cite{cagalj05,konorski06}, in the design of the algorithm we assume that all stations are saturated (i.e., always have a packet ready for transmission), they are in the transmission range of each other (i.e., no hidden nodes) and use the same modulation-coding scheme. In the simulations section, we show that the proposed algorithm can be extended to effectively prevent selfish behaviors with non-saturated stations. While the design assumes no hidden terminals, the algorithm also works for hidden terminals as long as the RTS/CTS mechanism is used. Furthermore, in case of different modulation-coding schemes, the algorithm can be applied to enforce the target configuration proposed in~\cite{winet}.

\subsection{Computation of $CW_{opt}$}

We use the model of \cite{tmc1} to compute the throughput $r_i$ of station $i$ as 
\small
\begin{equation}
r_i(\boldsymbol{\hat{\tau}}) = \frac{l}{T_s(\boldsymbol{\hat{\tau}}) }\hat{\tau}_i \prod_{j \neq i} \left(1-\hat{\tau}_j\right)\\
= \frac{\hat{\tau}_i}{1- \hat{\tau}_i}\frac{l}{T_s(\boldsymbol{\hat{\tau}}) }\prod_{j=1}^n \left(1-\hat{\tau}_j\right) 
\label{eq-r}
\end{equation}
\normalsize
where $\boldsymbol{\hat{\tau}}=[\hat{\tau}_1,\cdots,\hat{\tau}_n]$ are the probabilities that a station transmits in a slot time, $n$ is number of active stations in the WLAN, $l$ is the packet length in bits, $T_s(\boldsymbol{\hat{\tau}})  = T_t + \left(T_e - T_t\right)\prod_j\left(1-\hat{\tau}_j\right)$ is the average duration of a slot time in seconds, $T_t$ the duration of a transmission and $T_e$ the duration of an empty time slot.

By \cite[Lemma 1]{convexity}, the rate region boundary is the set of throughput vectors such that $\sum_{i=1}^nT_{air,i}(\boldsymbol{\hat{\tau}}) =1$ where $T_{air,i}(\boldsymbol{\hat{\tau}}) = \hat{\tau}_i\frac{T_t}{T_s(\boldsymbol{\hat{\tau}}) }$ is the fraction of airtime (including both successful and colliding transmissions) used by station $i$.  When all stations use the same transmission probability, it follows immediately that the value $\tau_{opt}$ maximising throughput is the unique solution to 
\begin{align}
\frac{1-n\tau_{opt}}{  \left(1-\tau_{opt}\right)^n }=1-\frac{T_e}{T_t}
\end{align}
Once we have $\tau_{opt}$ then $CW_{opt} = \frac{2}{\tau_{opt}} - 1$.   When $\frac{T_e}{T_t}$ is small, an accurate approximation is  $CW_{opt} = n\sqrt{\frac{2T_t}{T_e}} - 1$.

The following fundamental property will also prove useful:
\begin{theorem}\label{lem:five}
Consider the ball $C_{\Delta}=\{\boldsymbol{\hat{\tau}}: \boldsymbol{\hat{\tau}}\in[\tau_{opt} - \Delta , \tau_{opt} + \Delta]^n\}$ around $\tau_{opt}$, with $0\le\tau_{opt}-\Delta < \tau_{opt}+\Delta \le 1$ and $n \ge 2$.  For any $\boldsymbol{\hat{\tau}}\in C_{\Delta}$ the following inequality holds:
\begin{equation}
D(\boldsymbol{\hat{\tau}}):=nr_{opt} - \sum_j r_j(\boldsymbol{\hat{\tau}}) \le n\rho\Delta
\label{eq-F_bound_3}
\end{equation}
where $\rho=\frac{r_{opt}}{\tau_{opt}(1-\tau_{opt})}$, $r_{opt}$ is the maximum achievable throughput of a station when $\hat{\tau}_i = \hat{\tau}_j \ \forall i,j$ and $\tau_{opt}$ is the value of the transmission probability that leads to this throughput.
\end{theorem}
\emph{Proof:} See Appendix. 

This theorem bounds the difference $D(\boldsymbol{\hat{\tau}})$ between the optimum and actual WLAN sum-rate throughput.

\subsection{Algorithm description}

Following the theory of repeated games \cite{repeated}, GAS implements an \emph{adaptive algorithm} in which each station updates its $CW$ at every stage, while keeping the configuration of the other parameters fixed to the values provided in Section \ref{sec-objectives}.\footnote{Following the  802.11e standard, which updates the configuration of the 802.11e parameters at every beacon frame, we set the duration of a stage equal to the duration of a beacon interval. While the beacon interval can be set to different values, it is typically set to 100 ms.} The central idea behind GAS is that, when a station is detected as misbehaving, the other stations reduce their $CW$ in subsequent stages to prevent this selfish station from benefiting from misbehaving.


A key challenge in GAS is to carefully adjust the reaction against a selfish station. Indeed, as mentioned in the introduction, if the reaction is not severe enough a selfish station may benefit from its misbehavior, but if the reaction is too severe the system may become unstable by entering an endless loop where all stations indefinitely reduce their $CW$ to punish each other. In order to address this challenge, we design GAS using techniques from Lyapunov theory \cite{Lyapunov} that prevent the system from entering into a spiral of increasing punishments that lead to throughput collapse and guarantee that the $CW$ of all stations converges to $CW_{opt}$.

The iterative GAS update of the $CW$ values can be modeled as a discrete time dynamical system whose state is given by $\boldsymbol\tau = \left[\tau_1, \tau_2, \ldots, \tau_n\right]$, where $\tau_i$ is related to the probability with which station $i$ transmits in a slot time. That is:
\begin{equation}\label{eq-adaptive}
\boldsymbol\tau \left(t+1\right) = f \left( \boldsymbol\tau \left(t\right) \right)
\end{equation}
where $f : \mathbb{R}^n \rightarrow \mathbb{R}^n$ is a non-linear function that models the system dynamics.  The main design challenge is to determine the function $f$.  To this end, we adopt a standard feedback approach \cite{book} and update $\tau_i$ at each stage as:
\begin{equation}
\tau_i\left(t+1\right) = \tau_i\left(t\right) + \gamma g_i\left(\boldsymbol{\tau}(t)\right),\ i=1,\cdots,n
\label{eq-tau_updt}
\end{equation}
where $\gamma>0$ is a scalar parameter and $g_i: [0,1]^n \rightarrow [0,1]$ .

In order to allow for larger values of $\gamma$, which reduces the convergence time of the algorithm,\footnote{The fact that imposing a lower bound on $\tau_i$ allows for larger $\gamma$ values can be seen from the proof of Theorem \ref{th:one}.} we impose that the probability of transmitting in a slot time does not fall below $\tau_{opt}/2$. Similarly, if $\tau_i(t)$ exceeds 1, we transmit with probability 1. Thus,
\begin{equation}\label{eq-tauhat}
\hat{\tau}_i(t) = \min(1,\max(\tau_i(t),\tau_{opt}/2))
\end{equation}
where $\hat{\tau}_i(t)$ is the probability that the station $i$ transmits in a slot time after imposing the above constraints. Given  $\hat{\tau}_i(t)$, the $CW$ parameter of station $i$ at stage $t$ is $CW_i(t) = \frac{2}{\hat{\tau}_i(t)} - 1$.

We next address the design of function $g_i$ in (\ref{eq-tau_updt}). 
Our requirements when designing $g_i$ are twofold: ($i$) selfish stations should not be able to obtain extra throughput from the WLAN by following a different strategy from GAS, and ($ii$) as long as there are no selfish stations that deviate from GAS, the $\tau_i$ of all stations should converge to the optimal value $\tau_{opt}$. To meet the above requirements
we select $g_i$ as follows
\begin{equation}\label{eq-ei}
g_i(\boldsymbol{\tau}) = \sum_{j \neq i}{\left(r_j(\boldsymbol{\tau}) - r_i(\boldsymbol{\tau})\right)} - F_i(\boldsymbol{\tau})
\end{equation}
where $F_i(\boldsymbol{\tau})$ is a function that we design below.  Observe that $g_i(\boldsymbol{\tau})$ consists of the following two components, each of which fulfills one of the requirements identified above: 
\begin{itemize}
\item The first component, $\sum_{j\neq i}{r_j(\boldsymbol{\tau}) - r_i(\boldsymbol{\tau})}$, serves to punish selfish stations: if a station $i$ receives less throughput than the other stations, this component will be positive and hence station $i$ will increase its transmission probability $\tau_i$ to punish the other stations.
\item The second component, $F_i(\boldsymbol{\tau})$, drives the system to the target configuration in the absence of selfish behavior.
\end{itemize}

Regarding $F_i(\boldsymbol{\tau})$, to drive $\tau_i$ to the target value $\tau_{opt}$ we require $F_i(\boldsymbol{\tau})$ to be positive when $\tau_i > \tau_{opt}$, and negative otherwise.  Furthermore, $F_i(\boldsymbol{\tau})$ should not allow selfish stations to obtain a throughput gain over well-behaved stations.   To gain insight, we first consider steady-state operation, which implies that selfish stations play with a static configuration, and consider the case when all stations but a selfish one implement GAS. (In the analysis of Section \ref{sec-game} we show that GAS is also effective against selfish strategies that change the configuration over time.)   In steady-state the LHS and RHS of update (\ref{eq-ei}) must be equal for those stations using GAS, i.e. $g_i(\boldsymbol{\tau}^\infty)=0$ $\forall i\ne s$ and so 
\begin{align}
F_i(\boldsymbol{\tau}^\infty) =  \sum_{j \neq i}{\left(r_j(\boldsymbol{\tau}^\infty) - r_i(\boldsymbol{\tau}^\infty)\right)} =  r_s(\boldsymbol{\tau}^\infty)-r(\boldsymbol{\tau}^\infty)
\end{align}
where $s$ is the selfish station, $r(\boldsymbol{\tau}^\infty)$ is the throughput of a well-behaved station (which, by symmetry, is the same for all such stations in steady-state) and the $\infty$ superscript indicates values when the system is in steady state.   We require that the throughput of a selfish station does not exceed the target throughput, $r_s(\boldsymbol{\tau}^\infty) \leq r_{opt}$.   That is,

\small
\begin{align}
nr_s(\boldsymbol{\tau}^\infty) &= r_s(\boldsymbol{\tau}^\infty) + (n-1)r(\boldsymbol{\tau}^\infty) + (n-1)(r_s(\boldsymbol{\tau}^\infty) -r(\boldsymbol{\tau}^\infty)) \nonumber\\
&= \sum_{j=1}^n r_j(\boldsymbol{\tau}^\infty) + (n-1)F_i(\boldsymbol{\tau}^\infty)   \le n r_{opt}
\end{align}
\normalsize
which is satisfied when
\begin{align}
F_i(\boldsymbol{\tau}) &\leq \frac{1}{n-1}D(\boldsymbol{\tau}) \label{eq-climit}
\end{align}
where
\begin{equation}\label{eq-e}
D(\boldsymbol{\tau}) = nr_{opt} - \sum_j{r_j}(\boldsymbol{\tau})
\end{equation}
The intuition here is that when a selfish station misbehaves, it receives more throughput than the well-behaved stations. This, however, moves the point of operation away from the optimal one, reducing the overall efficiency in terms of the aggregate throughput. The bound (\ref{eq-climit}) ensures that the additional throughput received by the selfish station does not outweigh the throughput it loses due to the overall loss of aggregate throughput. This guarantees that in steady-state the selfish station does not receive more throughput and hence does not benefit from misbehaving.

Following the above requirements, we select $F_i(\boldsymbol{\tau})$ as:
\begin{equation}\label{eq-F}
F_i(\boldsymbol{\tau}) = \left\{
\begin{array}{l c}
D(\boldsymbol{\tau})/n, & \tau_i > \tau_{opt} \ \& \ D(\boldsymbol{\tau}) \geq 0\\
-D(\boldsymbol{\tau})/n, & \tau_i \leq \tau_{opt} \ \& \ D(\boldsymbol{\tau}) \geq 0 \\
D(\boldsymbol{\tau})/(n-1), & D(\boldsymbol{\tau}) < 0 \\
\end{array}
\right.
\end{equation}
This choice meets the design requirements set above for the function $F_i(\boldsymbol{\tau})$: ($i$) it satisfies (\ref{eq-climit}), preventing selfish stations from obtaining any gain; and ($ii$) for well-behaved stations it fulfills $F_i(\boldsymbol{\tau}) > 0$ for $\tau_i > \tau_{opt}$ and $F_i(\boldsymbol{\tau}) < 0$ for $\tau_i < \tau_{opt}$, which drives the system to optimal operation. The only exception is when $D(\boldsymbol{\tau})<0$: in this case (\ref{eq-climit}) imposes $F_i(\boldsymbol{\tau}) < 0$ independent of the value of $\tau_i$. However, as we show later in Section \ref{sec-stability}, this does not affect the convergence of the algorithm to the desired point of operation.

To compute update (\ref{eq-tau_updt}) with these choices of $g_i$ and $F_i$, each station only needs to measure at the end of every stage the throughput it has received during this stage as well as the throughput that each of the other stations has received.\footnote{Similarly to \cite{cagalj05,konorski06}, we rely on the broadcast nature of the wireless medium which provides WLAN stations with the ability to measure the throughput received by the other stations.}

\section{Algorithm Analysis}\label{sec-analysis}


In this section we study analytically the performance of the system. First, we prove that when all the stations are well-behaved and implement the GAS algorithm, the WLAN converges to the optimal configuration (Section~\ref{sec-stability}). Then, we show that a selfish station does not have any incentive to deviate from the GAS algorithm (Section~\ref{sec-game}).

\subsection{Stability Analysis}\label{sec-stability}
We show that when all stations implement GAS, the WLAN is driven to the optimal configuration, i.e.,~$\tau_i = \tau_{opt} \; \forall i$. 

Formally, a point $\boldsymbol{\tau}_e \in [0,1]^n$ is a said to be a \emph{globally asymptotically stable equilibrium} of the system (\ref{eq-adaptive}) if ($i$) $\forall \epsilon > 0 \ \exists \delta > 0 $ such that $\|\boldsymbol{\tau}\left(0\right) - \boldsymbol{\tau}_e\| < \delta \Rightarrow \|\boldsymbol{\tau}\left(t\right) - \boldsymbol{\tau}_e\| < \epsilon \ \forall t$; and ($ii$) $\lim_{t \rightarrow \infty} \boldsymbol{\tau}\left(t\right) = \boldsymbol{\tau}_e \ \forall \boldsymbol{\tau} \left(0\right) \in [0,1]^n$. These conditions ensure that the system converges to $\boldsymbol{\tau}_e$ independently of its initial state and that the equilibrium point is unique.   We have the following result:
\begin{theorem}[Global stability]\label{th:one}
The target configuration $\boldsymbol\tau_{opt}=[\tau_{opt}\cdots \tau_{opt}]^T$ is a globally asymptotically stable equilibrium point under update (\ref{eq-tau_updt}) provided 
\begin{equation}
\gamma < \gamma_{max} := \left(\frac{nl}{T_{opt}}\left(1-\frac{\tau_{opt}}{2}\right)^{n-2}\right)^{-1} \nonumber
\end{equation}
and $n \ge 2$.
\end{theorem}
\emph{Proof:} See Appendix. 

The proof of Theorem \ref{th:one} makes use of Lyapunov's direct method \cite{Lyapunov}.  Namely, a point $\boldsymbol{\tau}_e$ is the globally asymptotically stable equilibrium of the system if there exists a continuous radially unbounded function $V: \mathbb{R}^n \rightarrow \mathbb{R}$ such that $V\left(\boldsymbol{\tau} - \boldsymbol{\tau}_e\right) > 0 \; \forall \boldsymbol{\tau} \neq \boldsymbol{\tau}_e$, $V\left(\boldsymbol{\tau}_e\right) = 0$ and
\begin{equation}
V\left(\boldsymbol{\tau}\left(t+1\right) - \boldsymbol{\tau}_e \right) < V\left(\boldsymbol{\tau}\left(t\right) - \boldsymbol{\tau}_e \right)
\label{Lyapunov}
\end{equation}
To apply this result we must find a Lyapunov function $V$ with these properties.   Selecting  $V\left(\boldsymbol{\tau} - \boldsymbol{\tau}_e\right) = \|\boldsymbol{\tau} - \boldsymbol{\tau}_e\|_\infty$ as a candidate Lyapunov function, with the equilibrium point $\boldsymbol{\tau}_e= \boldsymbol\tau_{opt}=[\tau_{opt}\cdots \tau_{opt}]^T$, the proof of Theorem \ref{th:one} establishes that 
\begin{equation}
\|\boldsymbol\tau\left(t+1\right) - \boldsymbol\tau_{opt}\|_\infty < \|\boldsymbol\tau\left(t\right) - \boldsymbol\tau_{opt}\|_\infty
\label{Lyapunov_2}
\end{equation} 
That is, Lyapunov equation (\ref{Lyapunov}) is satisfied by this choice of $V$ and so $\boldsymbol\tau_{opt}$ is a globally asymptotically stable equilibrium.

It remains to select the value of the parameter $\gamma$.  This involves a tradeoff: the smaller $\gamma$, the slower the rate of convergence is; however, if $\gamma$ is set too large the system risks instability (as shown by Theorem \ref{th:one}). Following the same rationale as the \emph{Ziegler-Nichols} rules \cite{franklin}, which have been proposed to address a similar tradeoff in the context of classical control theory, we recommend setting $\gamma$ to half of the value at which the system turns unstable, i.e., $\gamma = \gamma_{max}/2$.

\subsection{Game Theoretic Analysis}\label{sec-game}

In the previous section we have seen that, when all stations implement GAS, the system converges to the target configuration, i.e., all stations play with $\tau_i = \tau_{opt}$ and receive a throughput equal to $r_{opt}$. In this section we conduct a game theoretic analysis to show that a selfish station cannot obtain more throughput than $r_{opt}$ by following a different strategy from GAS. 
In what follows, we say that a station is \emph{honest} or well-behaved when it implements GAS to configure its 802.11e parameters, while we say that it is \emph{selfish} or misbehaving when it plays a different strategy from GAS to configure its parameters with the aim of obtaining some gain.

The game theoretic analysis conducted in this section assumes that users are \emph{rational} and want to maximize their own benefit or \emph{utility}, which is given by the throughput. The model is based on the theory of \emph{repeated games} \cite{repeated}. With repeated games, time is divided into stages and a player can take new decisions at each stage based on the observed behavior of the other players in  previous stages. This matches our algorithm, where time is divided into intervals and stations update their configuration at each interval. Like other previous analyses on repeated games \cite{cagalj05,konorski06}, we consider an infinitely repeated game, which is a common assumption when the players do not know when the game will finish. Using this model, the following theorem shows the effectiveness of GAS against a selfish station. Note that the theorem does not impose any restriction on the strategy followed by the selfish station, which may play with all the four 802.11e parameters changing their setting over time.

\begin{theorem}\label{th:two}  Let us consider a selfish station that uses a configuration that can vary over time. If all the other stations implement the GAS algorithm, the throughput received by this station will be no larger than $r_{opt}$.
\end{theorem} 
\emph{Proof:} See Appendix. 

\newtheorem{corollary}{Corollary}
\begin{corollary}\label{cor:one} All-GAS is a Nash equilibrium of the game.
\end{corollary}\begin{proof}
By Theorem \ref{th:two}, if all other stations play GAS, then the best response of this station is to play GAS as well since it cannot benefit from playing a different strategy. Thus, \emph{All-GAS} is a Nash equilibrium.
\end{proof}

This shows that, if all stations start playing with no previous history, then none of them can gain by deviating from GAS. In addition to this, in repeated games it is also important to make sure that, if at some point the game has a given history, a selfish station cannot take advantage of this history to obtain any gain by playing a different strategy from GAS. The following theorem confirms that \emph{All-GAS} is a Nash equilibrium of any subgame (where a \emph{subgame} is defined as the game resulting from starting to play with a certain history). Therefore a selfish station cannot benefit by deviating from GAS independently of the previous history of the game.

\begin{theorem}\label{th:three} 
All-GAS is a subgame perfect Nash equilibrium of the game.
\end{theorem}
\emph{Proof:} See Appendix. 

\section{Performance Evaluation}\label{sec-performance}

In this section we thoroughly evaluate GAS by conducting an extensive set of simulations to show that $(i)$ a selfish station cannot benefit from following a different strategy from GAS, and $(ii)$ when all stations are well-behaved, GAS provides optimal performance, is stable and reacts quickly to changes. For the simulations, we have implemented our algorithm in OMNET++ ({\ttfamily\url{www.omnetpp.org}}). The physical layer parameters of IEEE 802.11g and a fixed payload size of 1500 bytes have been used in all the experiments. In the simulations of Sections \ref{sec-firstsim} to \ref{sec-perturbations}, we focus on the $CW$ parameter: we assume that all stations (both honest and selfish) use a fixed configuration of their $AIFS$, $TXOP$ and $m$ parameters equal to the target configuration and play only with the $CW$ parameter. Then, in the simulations of Section~\ref{sec-parameters} we study all the four parameters and show that selfish stations cannot obtain any benefit from any configuration of these parameters. Unless otherwise stated, we assume that all stations are sending traffic under saturation conditions. For all simulation results, 95\% confidence intervals are below 0.5\%.

\subsection{Selfish station behavior}\label{sec-firstsim}

In order to gain insight into the impact of the behavior of a selfish station with the GAS algorithm, we evaluate the resulting throughput distribution when a selfish station uses a fixed $CW_i$ and all other stations implement GAS. Figure~\ref{fig:constantcw} shows the throughput obtained by the selfish station and an honest one as a function of the $CW_i$ used by the selfish station, $CW_s$, when there are $n = 10$ stations in the WLAN. We observe from this figure that there are some $CW_s$ values for which the selfish station obtains a larger throughput than the honest one, and others for which the honest stations obtain a larger throughput. However, as long as the selfish station plays with the $CW_s$ value that maximizes its throughput, it does not receive more throughput than the honest station, and hence does not have any gain over an honest station as a result of its selfish behavior.

\begin{figure}
\includegraphics[width=25em, angle=0]{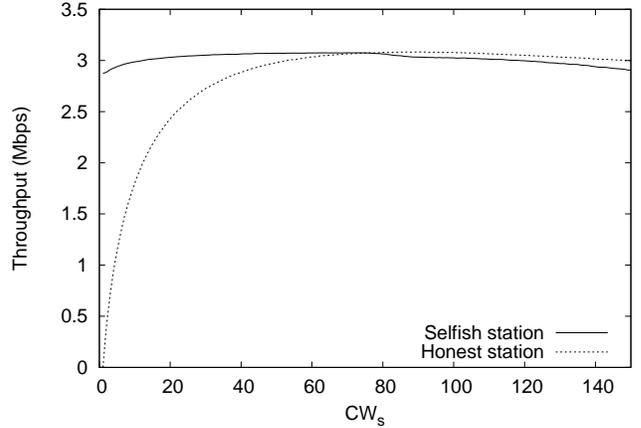}
\caption{Throughput of a selfish and a honest station as a function of the $CW_i$ of the selfish station.} 
\label{fig:constantcw}
\end{figure}

\subsection{Protection against selfish stations}

According to the game theoretic analysis conducted in Section \ref{sec-game}, a station cannot obtain more throughput with a selfish strategy than by playing GAS. To validate this result, we evaluate the throughput obtained by a selfish station with the following strategies. In the first strategy (\emph{static}), the selfish station uses the fixed configuration of $CW_{i}$ that provides the largest throughput, obtained from performing an exhaustive search over all possible configurations (like in Figure \ref{fig:constantcw}). In the second strategy (\emph{adaptive 1}), the selfish station periodically tries $CW_i = 2$ to gain throughput and when it realizes that its throughput is below $r_{opt}$, it assumes that it has been detected as selfish and switches back to $CW_i = CW_{opt}$. The third strategy (\emph{adaptive 2}) is similar to the previous one but instead of switching back to $CW_{opt}$, the station increases its $CW_i$ by 5. In the last strategy (\emph{adaptive 3}), the selfish station decreases its $CW_i$ by 5 as long as its throughput is larger than in the previous stage and increases it by 5 otherwise. Figure~\ref{fig:variablecw} compares the throughput obtained with each of these strategies against that obtained with GAS for different $n$ values. We observe that, when all other stations play GAS, a given station maximizes its payoff by playing GAS, as it obtains a larger throughput with GAS than with any of the other strategies. This confirms the result of Theorem~2.

\begin{figure}
\includegraphics[width=17em, angle=270]{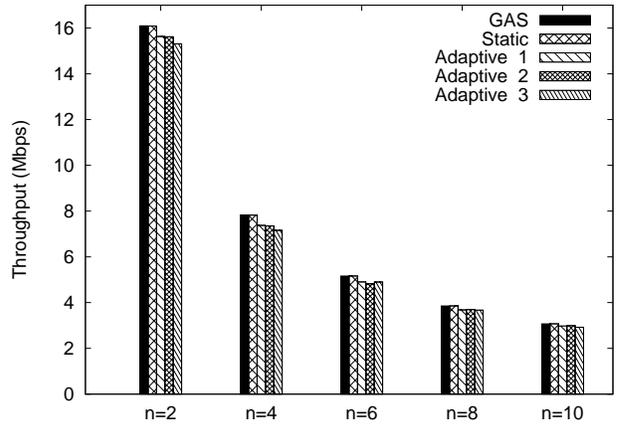}
\caption{Throughput of a station with different strategies.} 
\label{fig:variablecw}
\end{figure}

\subsection{Throughput performance}

The GAS algorithm has been designed with the goal of optimizing throughput performance when all stations play GAS. To verify this goal, we evaluate the throughput performance as a function of the number of stations $n$ when all stations play GAS. As a benchmark against which to compare the throughput performance, we consider a WLAN in which the $CW_i$ of all stations is statically set to the optimal value $CW_{opt}$. The results from the above experiment are illustrated in Table \ref{tab:validation}. We observe that the throughput performance resulting from GAS follows very closely the optimal configuration. Based on this, we conclude that the proposed algorithm is effective in providing optimal throughput performance.

\begin{table}
\begin{center}
\caption{Throughput per Station (Mbps)}
\begin{tabular}{|c|c|c|}
\hline
$n$ & \emph{All-GAS} & \emph{All-}$CW_{opt}$ \\
\hline
\hline
4 & 7.82 & 7.83 \\
8 & 3.85 & 3.86 \\
12 & 2.56 & 2.56 \\
16 & 1.91 & 1.92 \\
20 & 1.53 & 1.53 \\
\hline
\end{tabular}
\label{tab:validation}
\end{center}
\end{table}

\subsection{Stability and speed of reaction}

To validate that our system guarantees a stable behavior, we analyze the evolution over time of the parameter $CW_i$ for our $\gamma$ setting and a configuration of this parameter 10 times larger, in a WLAN with 10 stations. We observe from Figure \ref{fig:stability} that with the proposed configuration (label ``$\gamma$''), the $CW_i$ only presents minor deviations around its stable point of operation, while if a larger setting is used (label ``$10\gamma$''), the $CW_i$  has a strong unstable behavior with drastic oscillations. 

\begin{figure}
\includegraphics[width=25em, angle=0]{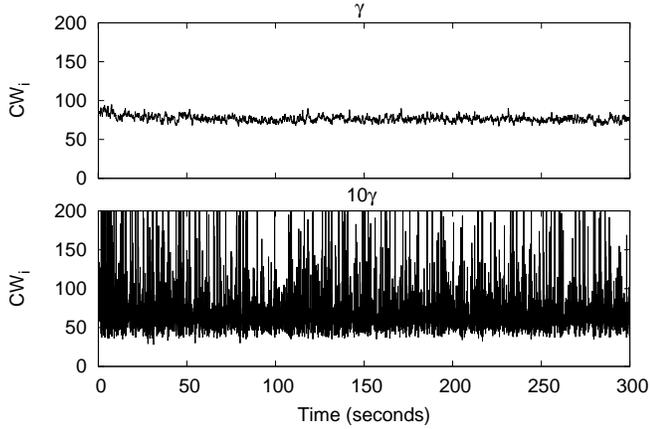}
\caption{System stability for different $\gamma$ settings.}
\label{fig:stability}
\end{figure}

To investigate the speed with which the system reacts against a selfish station, we consider a WLAN with 10 stations where initially all stations play GAS and then, after 50 seconds, one station changes its $CW_i$ to 2. Figure \ref{fig:reaction} shows the evolution of the throughput of the selfish station over time. We observe from the figure that with our setting (label ``$\gamma$''), the system reacts quickly, and in less than a few tens of seconds the selfish station does no longer benefit from its misbehavior. In contrast, for a setting of this parameter 10 times smaller (label ``$\gamma/10$''), the reaction is very slow and even 5 minutes afterwards, the selfish station is still receiving about 1 Mbps extra throughput.

Since with a larger setting of $\gamma$ the system suffers from instability while with a smaller one it reacts too slowly, we conclude that the proposed setting provides a good tradeoff between stability and speed of reaction.

\begin{figure}
\includegraphics[width=25em, angle=0]{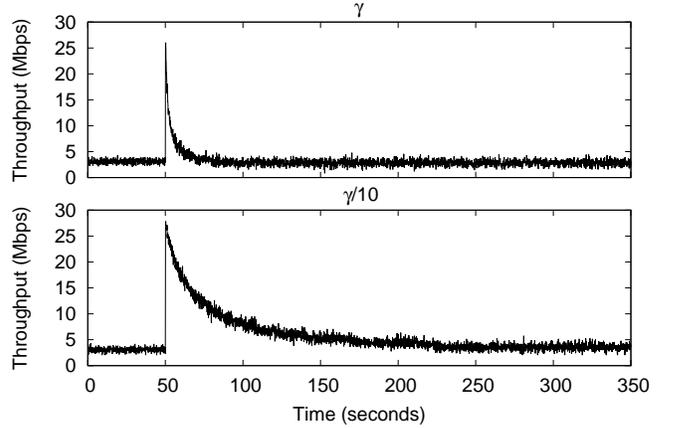}
\caption{Speed of reaction for different $\gamma$ settings.} 
\label{fig:reaction}
\end{figure}

\subsection{Comparison against other approaches}\label{sec-comparison}
In order to illustrate the advantages of GAS over other approaches, we compare the performance of GAS against CRISP \cite{konorski06} and the standard DCF configuration when there is a selfish station in the WLAN. In particular, we consider a WLAN with a selfish station that plays with the $CW_i$ value that maximizes its payoff and show the throughput received by an honest station.\footnote{To obtain the $CW_i$ that maximizes the selfish station's throughput, we have evaluated all possible $CW_i$ values and chosen the one that provides the largest throughput to the selfish station.}

The results, depicted in Figure~\ref{fig:comparison}, show that GAS outperforms very substantially CRISP and DCF. Since CRISP has been designed to prevent only extremely selfish users with $CW_i = 2$ or $CW_i = 1$, a selfish user with a slightly larger $CW_i$ goes undetected and can gain very significant throughput, leaving honest stations with very low throughputs as shown in the figure. With the standard DCF configuration, a selfish station maximizes its gain with $CW_i=1$, which yields zero throughput for the honest stations.

\begin{figure}
\includegraphics[width=25em, angle=0]{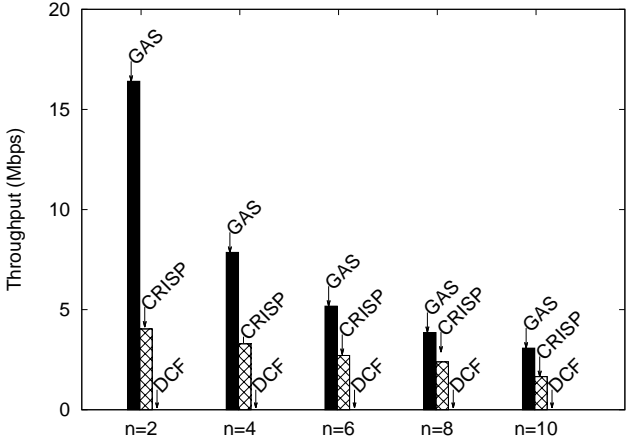}
\caption{Comparison against CRISP and DCF.} 
\label{fig:comparison}
\end{figure}

%

\subsection{Robustness to perturbations}\label{sec-perturbations}
One of the goals in the design of GAS has been its robustness against any kind of perturbation. Indeed, as it has been proved by Theorem \ref{th:one}, our system is guaranteed to converge to the desired point of operation independent of the initial state. Therefore, no matter the state to which the system is brought by a perturbation, it will always be able to recover.

In order to show the above feature, we consider the following experiment. In a WLAN with 15 stations, all running GAS, we introduce a burst of errors that affect one of the stations during one second. Figure~\ref{fig:perturbations} shows the evolution of the throughput of the affected station and one of the other stations over time. The figure also shows the behavior provided by CRISP under the same conditions. 

We observe from the figure that GAS quickly converges to the desired point of operation after the perturbation. Indeed, right after the perturbation the station that suffered the burst of errors believes that the other stations are behaving selfishly (as they have received a larger throughput) and plays with a smaller $CW_i$ for a while, which results in a larger throughput for this station. However, after a short transient all stations go back to playing with the optimal $CW_{opt}$.

In contrast to the above, CRISP does not show a robust behavior. With CRISP, the affected station plays $CW_i = 2$ after the burst to punish the others. These react by decreasing their $CW_i$ to 2 and eventually to 1, and from this point on stations keep punishing each other which brings the total throughput in the WLAN practically to 0. The WLAN remains in this state for the rest of the simulation run, which is 300 seconds long (only the first 20 seconds are shown in the graph).

\begin{figure}
\includegraphics[width=25em, angle=0]{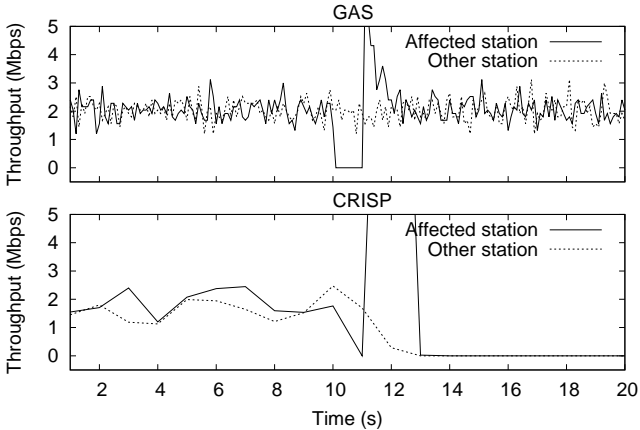}
\caption{Robustness to Perturbations of GAS vs. CRISP.}
\label{fig:perturbations}
\end{figure}

\subsection{Other 802.11e parameters}\label{sec-parameters}
In the experiments so far, we have considered that selfish users only play with the $CW_i$ parameter. 
However, according to the 802.11e standard, there are a number of additional parameters a user can play with, namely the backoff stage $m$, the arbitration interframe space $AIFS$ and the transmission opportunity $TXOP$. In order to show that a selfish user cannot benefit from playing with any of these parameters, we have conducted a number of experiments in which the parameters are set to different values from the ones given in Section~\ref{sec-objectives}. 
For each of the settings considered for these parameters, the selfish station uses the $CW_{min}$ that maximizes its throughput.

The results of the above experiment are given in Figure~\ref{fig:parameters} for different $n$ values. We observe that the selfish station never obtains any gain by deviating from GAS independent of the parameters it plays with. We conclude that GAS is effective not only against the $CW_i$ parameter but also against all the other configurable parameters of the 802.11e standard. This confirms the result of Theorem \ref{th:two}, according to which a station cannot benefit from following a selfish strategy to configure (all of) its parameters. This result is particularly relevant since the previous approaches \cite{cagalj05,konorski06} focus only on the $CW_i$ parameter and are not evaluated against any other parameter.

\begin{figure}
\includegraphics[width=25em, angle=0]{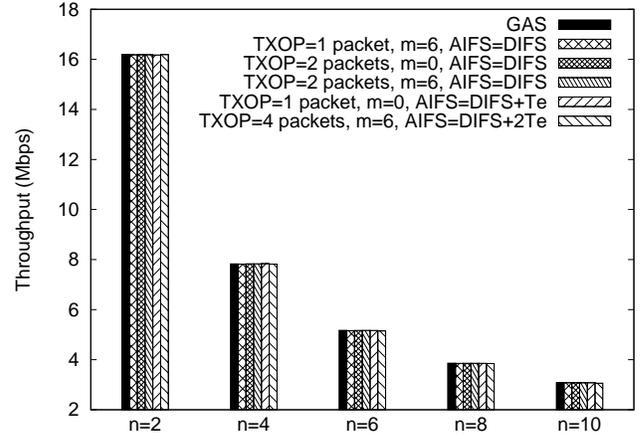}
\caption{Protection of GAS against selfish strategies with different 802.11e parameters.} 
\label{fig:parameters}
\end{figure}

\subsection{Non-saturated stations}

So far we have assumed that all stations are saturated, 
which is the most relevant case for selfishness and the only one considered in \cite{cagalj05,konorski06}. However, GAS can be easily extended to support non-saturated stations as follows: ($i$) to avoid reacting upon other stations receiving more throughput, a non-saturated station does not use the GAS algorithm to compute its configuration; ($ii$) a saturated station only includes in the sum of (\ref{eq-ei}) those stations that are receiving more throughput, thus excluding the non-saturated stations; and ($iii$) to compute $CW_{opt}$, we take into account the sending rate of the non-saturated stations (following, e.g., \cite{tvt}). To show the performance of GAS with non-saturated stations, we consider a WLAN with 10 stations, half of them saturated and the other half sending at a rate equal to half of the saturation throughput. In this scenario, a station that behaves selfishly using the $CW_i$ value that maximizes its gain obtains a throughput of 4.51 Mbps, while it would obtain 4.52 Mbps if it played GAS. This confirms the effectiveness of the algorithm in thwarting selfish behaviors in presence of non-saturated stations.

\section{Experimental Prototype}\label{sec-experiment}

One of the advantages of GAS is that it relies on functionality readily available in standard devices and therefore can be implemented with current off-the-shelf hardware. In order to show this, we have implemented our algorithm on Linux-based laptops. In this section, we report the experience gained from this prototype.

Our implementation is based on Linux kernel 2.6.24 laptops equipped with Atheros AR5212 cards operating in 802.11a mode and employing the MadWifi v0.9.4 driver. 
The GAS algorithm runs as a user-space application. In order to collect information about other stations' throughput, GAS uses a virtual device configured in promiscuous mode and monitors all frames that belong to the same BSS. With this information, it computes the $CW$ configuration by executing the algorithm described in Section \ref{sec-algorithm} and updates the computed $CW_{min}$ and $CW_{max}$ parameters in the driver every beacon interval by means of a private {\ttfamily IOCTL} call.

In order to validate our implementation, we deployed a small testbed consisting of three laptops, two of them sending traffic to the third one. For the traffic generation, nodes ran the {\ttfamily iperf} 
tool to generate 1470 byte UDP packets. The sending rate at each station was set to 20~Mbps, ensuring that they always had a packet ready for transmission. 
With this setting, we ran 
the experiment and measured the resulting throughput performance as a function of the configuration strategy followed by each of the two sending stations. In particular, we considered the following strategies: $(i)$ both stations employing the GAS algorithm to compute their configuration $(ii)$ both stations using a fixed $CW$ configuration (for a wide range of $CW$ values), and $(iii)$ one station executing GAS and the other one using a fixed $CW$ configuration.

The results of the above experiments are provided in Table~\ref{tab:implem}. For each experiment, the average throughput and the standard deviation of 5 runs of 300 seconds each is given. From these results, we draw the following conclusions: ($i$) when all stations are well-behaved, GAS not only achieves a fair throughput allocation but also outperforms any fair static $CW_i$ configuration, which shows that GAS is effective in providing \emph{optimal throughput performance}; and ($ii$) GAS is also effective against \emph{selfish configurations}, as shown by the fact that, when one station plays GAS, the other station is better off playing GAS than any other configuration.

\begin{table}
\begin{center}
\caption{Experimental Results}
\label{tab:implem}
\begin{tabular}{|c|c|c|c|}
\hline
Strategy & Strategy & Throughput & Throughput\\
station 1 & station 2 & station 1 & station 2\\
$(CW_1)$ & $(CW_2)$ & $(Mbps)$ & $(Mbps)$\\
\hline \hline
GAS & GAS & $14.91 \pm 0.07$ & $14.92 \pm 0.12$ \\ \hline
2 & 2 & $13.24 \pm 0.07$ & $13.21 \pm 0.07$\\
4 & 4 & $14.31 \pm 0.1$ & $14.31 \pm 0.12$\\
8 & 8 & $14.84 \pm 0.08$ & $14.87 \pm 0.08$\\
16 & 16 & $14.71 \pm 0.03$ & $14.68 \pm 0.04$\\
32 & 32 & $13.68 \pm 0.06$ & $13.72 \pm 0.06$\\
64 & 64 & $11.83 \pm 0.05$ & $11.82 \pm 0.06$\\
128 & 128 & $9.18 \pm 0.03$ & $9.18 \pm 0.03$\\ \hline
\multirow{7}{*}{GAS} & 2 & $12.44 \pm 0.01$ & $14.51 \pm 0.12$\\
& 4 & $13.83 \pm 0.11$ & $14.78 \pm 0.12$\\
& 8 & $14.89 \pm 0.06$ & $14.82 \pm 0.07$\\
& 16 & $15.05 \pm 0.4$ & $14.15 \pm 0.35$\\
& 32 & $14.74 \pm 0.16$ & $12.97 \pm 0.11$\\
& 64 & $14.03 \pm 0.63$ & $10.52 \pm 0.31$\\
& 128 & $12.87 \pm 0.37$ & $7.78 \pm 0.16$\\ \hline
\end{tabular}
\end{center}
\end{table}

\section{Conclusions}\label{sec-conclusions}

Following the 802.11e standard, which opens the configuration of a number of contention parameters of the MAC layer, current WLAN cards allow the modification of these parameters by means of a simple command. One of the problems raised by this functionality offered by WLAN cards is that users can selfishly configure the parameters used by their station to increase their share of throughput at the expense of the other users.
%
In order to prevent this undesirable behavior, in this paper we design a novel adaptive algorithm called GAS (\emph{Game-theoretic Adaptive Stable}). With the GAS algorithm, upon detecting a selfish station users react by using a more aggressive configuration of the parameters that serves to punish the selfish station.

A critical aspect in the design of such an adaptive algorithm is to carefully adjust the reaction against a selfish station to avoid that the system turns unstable by overreacting. 
%
%
By conducting a \emph{Lyapunov stability analysis} of the GAS algorithm, we show that, when all the stations in the WLAN run GAS, the system is globally stable and converges to the desired configuration. Furthermore, by conducting a game theoretic analysis based on \emph{repeated games}, we show that a selfish station cannot benefit from playing a different strategy from GAS (neither with a fixed configuration nor a variable one).

We exhaustively evaluate the performance of GAS by means of simulations and show that: $(i)$ GAS is effective in optimizing throughput performance, $(ii)$ it is also effective against selfish stations using a fixed or variable configurations, and $(iii)$ it outperforms other approaches both in terms of protection against selfish configurations and robustness. Additionally, GAS is validated by means of an experimental prototype, confirming that it can be implemented on commodity hardware.

%

\bibliographystyle{IEEEtran}
\bibliography{./gametheory}

\section*{Appendix}

\subsection{Proof of Theorem \ref{lem:five}}\label{A}


We proceed by establishing two useful Lemmas, and then present the proof of Theorem  \ref{lem:five}.

\begin{lemma}\label{lem:three}
Consider the set of points $C(p_e)=\{\boldsymbol{\hat{\tau}}: \boldsymbol{\hat{\tau}}\in[\tau_m,\tau_M]^n, \prod_{j=1}^n(1-\hat{\tau}_j)=p_e\}$, $0\le \tau_m <\tau_M\le 1$. Over set $C(p_e)$, the vector $\boldsymbol{\hat{\tau}}$ minimising $\sum_{i=1}^n r_i(\boldsymbol{\hat{\tau}}) $ has all elements equal, i.e. $\hat{\tau}_i=\hat{\tau}_j$, $i,j\in\{1,2,\cdots,n\}$.
\end{lemma}
\begin{proof}
By (\ref{eq-r}),
$
\sum_{i=1}^n r_i 
=\frac{l}{T_t}\frac{1}{\prod_{j=1}^n(1+x_j) +a}\sum_{i=1}^n x_i
$ 
where $x_i=\hat{\tau}_i/(1-\hat{\tau}_i)$ and $a=(T_e-T_t)/T_t$.  Minimising $\sum_{i=1}^n r_i $ over set $C(p_e)$ then corresponds to the following optimisation
\begin{align*}
\min_{x_i,i=1,2,\cdots,n}  \sum_{i=1}^n x_i \quad s.t. \prod_{j=1}^n(1+x_j)=\frac{1}{p_e}, x_m\le x_i\le x_M \forall i
\end{align*}
which we can rewritten equivalently as
\begin{align*}
\min_{z_i,i=1,2,\cdots,n}  \sum_{i=1}^n e^{z_i}-1 \ s.t. \sum_{j=1}^n z_i=\log\frac{1}{p_e}, z_m\le z_i\le z_M \forall i
\end{align*}
where $z_i = \log(1+x_i)$, $z_m=\log(1+x_m)$, $z_M=\log(1+x_M)$.   It is enough to show that any optimum $\mathbf{z}^*$ satisfies $z^*_i=z^*_j$, $i,j\in\{1,2,\cdots,n\}$.  

The objective is convex, the equality constraint is linear and the inequality constraints convex, hence this is a convex optimisation.  Since $\tau_m <\tau_M$ the Slater condition is satisfied and so strong duality holds.  The Lagrangrian $L$ is
\begin{minipage}{\columnwidth}
\footnotesize
\begin{align*}
\sum_{i=1}^n e^{z_i}-1 - \lambda(\sum_{j=1}^n z_i-\log\frac{1}{p_e}) 
+ \sum_{i=1}^n\underline{\theta}_i(z_m-z_i) 
 + \sum_{i=1}^n\bar{\theta}_i(z_i-z_M)
\end{align*}
\normalsize
\end{minipage}
and the main KKT conditions are
\begin{align}
\frac{\partial L}{\partial z_i}\bigg|_{z_i = z_i^*} &= e^{z^*_i} - \lambda - \underline{\theta}_i + \bar{\theta}_i = 0,\ i=1,2,\cdots,n
\end{align}
which must be satisfied by any optimal point $\mathbf{z}^*$.  When $\underline{\theta}_i  = 0 =\bar{\theta}_i$, $i=1,2,\cdots,n$ it follows from the KKT conditions that the minimum occurs when $z^*_i=z^*_j$, $i,j\in\{1,2,\cdots,n\}$.   

When $\underline{\theta}_i  > 0$ for some $i$ (and so by complementary slackness $z^*_i=z_m$), we would like to show that we must have $\underline{\theta}_i  > 0$ for all $i=1,2,\cdots,n$ (and so $z^*_i=z_m$ for all $i$).    Firstly, when $\log(1/p_e) = n z_m$, since $z^*_i\ge z_m$ it follows immediately that $z^*_i=z_m$ for all $i$.   Otherwise, we proceed by contradiction.  Suppose that $nz_m < \log(1/p_e) \le n z_M$ and $z^*_i=z_m$ for some $i$.   The elements of $\mathbf{z}^*$ are therefore not all the same value.  Consider the point $y_i =\log(1/p_e)/n$, $i=1,2,\cdots,n$.   This point satisfies the constraints
\begin{align}
\sum_{i=1}^n y_i = \log\frac{1}{p_e},\ z_m \le y_i \le  z_M
\end{align}
and so is feasible.  By the strict convexity of the exponential we have
$
ne^{\frac{1}{n}\sum_{i=1}^n z^*_i} = ne^{\log(1/pe)/n} < \sum_{i=1}^n e^{z^*_i} 
$ 
(with strict inequality since the elements of $\mathbf{z}^*$ are, by assumption, not all the same value).   Observing that $\sum_{i=1}^n e^{y_i} = ne^{\log(1/p_e)/n}$, it follows immediately that $\sum_{i=1}^n e^{y_i} < \sum_{i=1}^n e^{z^*_1}$ yielding a contradiction.  That is, we must either have $z^*_i=z_m$ for all $i$ or $z^*_i \ne z_m$ for all $i$. Since in the case we are analyzing we have $z^*_i=z_m$ for some $i$, this implies $z^*_i=z_m$ for all $i$. By an almost identical argument, it also follows that when $\underline{\theta}_i  < 0$ for some $i$, we must have $z^*_i=z_M$ for all $i$. 
\end{proof}

\begin{lemma}\label{lem:four}
Consider the set of points $C=\{\boldsymbol{\hat{\tau}}: \boldsymbol{\hat{\tau}}\in[\tau_m,\tau_M]^n\}$, $0\le \tau_m \le\tau_M\le 1$. Over set $C$, the vector $\boldsymbol{\hat{\tau}}$ minimising $\sum_{i=1}^n r_i(\boldsymbol{\hat{\tau}})$ satisfies either $\hat{\tau}_i=\tau_m$ for all $i$ or $\hat{\tau}_i=\tau_M$ for all $i$.
\end{lemma}
\begin{proof}
If $\tau_m =\tau_M$ the result follows trivially.  Suppose therefore $\tau_m <\tau_M$.  By Lemma \ref{lem:three}, minimising $\sum_{i=1}^n r_i$  s.t. $\tau_m \le \hat{\tau}_i \le \tau_M$, $i=1,2,\cdots,n$ is equivalent to setting $\hat{\tau}_i=\tau^*$, $i=1,2,\cdots,n$ and finding a $\tau^*$ solving
\begin{equation}
\min_{\tau_m\le \hat{\tau} \le \tau_M} \frac{n\hat{\tau} }{1-\hat{\tau}} \frac{(1-\hat{\tau})^n}{T_t + (T_e-T_t)(1-\hat{\tau})^n}l
\end{equation}

Taking logs and letting $x=\frac{\hat{\tau}}{1-\hat{\tau}}$, $\tilde{x} = \log x$ this optimisation can be rewritten as
$
\min_{\tilde{x}_m \le \tilde{x} \le \tilde{x}_M} \tilde{r}(\tilde{x})
$
with $\tilde{r}(\tilde{x}) = \tilde{x} - \log( (1+e^{\tilde{x}})^n +a ) + \log \frac{nl}{T_t}$, $a=(T_e-T_t)/T_t$, $\tilde{x}_m=\log\frac{\tau_m}{1-\tau_m}$, $\tilde{x}_M=\log\frac{\tau_M}{1-\tau_M}$.  Importantly, the objective function $\tilde{r}(\cdot)$ is concave in $\tilde{x}$, since ($i$) the first term is linear; ($ii$) expanding the $(1+e^{\tilde{x}})^n$ term, it can be verified that the second term is convex \cite{Boyd}; and ($iii$) the third term is constant. Hence, for any $\tilde{x} = \alpha \tilde{x}_m + (1-\alpha)\tilde{x}_M$, $0 \le \alpha \le 1$ lying in the interval $[\tilde{x}_m, \tilde{x}_M]$ we have $\tilde{r}(\tilde{x})  \ge \alpha \tilde{r}(\tilde{x}_m) + (1-\alpha) \tilde{r}(\tilde{x}_M)$.  It follows immediately that the minimum of $\tilde{r}(\tilde{x})$ over interval $[\tilde{x}_m, \tilde{x}_M]$ must be located at one of the boundary points.
\end{proof}


\begin{proof}[Proof of Theorem \ref{lem:five}]
%
Let us denote $\tau_m = \tau_{opt} - \Delta$ and $\tau_M = \tau_{opt} + \Delta$. By Lemma \ref{lem:four}, $r_{opt} - \frac{1}{n}\sum_j r_j$ is maximized either when $\hat{\tau}_i = \tau_{m} \ \forall i$ or $\hat{\tau}_i = \tau_{M} \ \forall i$. For $\hat{\tau}_i = \tau_{m} \ \forall i$, we have
\begin{equation}
\frac{1}{n}\sum_j r_j = \frac{\tau_m(1-\tau_m)^{n-1}l}{T_{s,m}} \geq \frac{\tau_m(1-\tau_{opt})^{n-1}l}{T_{opt}}
\end{equation}
where $T_{s,m}$ and $T_{opt}$ are the values of $T_s$ when $\tau_i = \tau_m \ \forall i$ and $\tau_i = \tau_{opt} \ \forall i$, respectively. From the above,
\begin{align*}
r_{opt} - \frac{1}{n}\sum_j r_j &\leq (\tau_{opt}- \tau_m) \frac{(1-\tau_{opt})^{n-1}l}{T_{opt}}\\
& = \Delta \frac{r_{opt}}{\tau_{opt}} \leq \Delta \frac{r_{opt}}{\tau_{opt}(1-\tau_{opt})}
\end{align*}
from which we have that (\ref{eq-F_bound_3}) holds for this case.

We next address the case $\hat{\tau}_i = \tau_M \ \forall i$. If $\tau_M > 2\tau_{opt}$, it is easy to see that (\ref{eq-F_bound_3}) holds, as in this case $\Delta > \tau_{opt}$ and thus $\rho n \Delta > n r_{opt}$. 
To prove that (\ref{eq-F_bound_3}) also holds for $\tau_M \leq 2\tau_{opt}$, we proceed as follows. Let $r(\tau)$ be the throughput of a station as a function of $\tau$ when $\hat{\tau}_i = \tau$ for all $i$. Then,
\begin{equation}\label{eq-derivative}
r_{opt} - r(\tau_M) = \int_{\tau_M}^{\tau_{opt}}{\frac{d r(\tau)}{d \tau} d\tau}
\end{equation}
with $\frac{\partial r(\tau)}{\partial \tau} = \frac{l(1-\tau)^{n-2}(T_s - n \tau T_t)}{T_s^2}$.  
We next show that the above derivative is negative in the interval $\tau \in [\tau_{opt}, \tau_{M}]$. The sign of the derivative depends on that of the term $T_s - n \tau T_t$. Since the throughput is maximized at $\tau_{opt}$ and $n>1$, the derivative at $\tau = \tau_{opt}$ is 0 (when the number of stations $n>1$ the optimum attempt probability must lie in the interior of $[0,1]^n$), and so $T_s - n \tau T_t$=0. The derivative of $T_s - n \tau T_t$ is $n(1-\tau)^{n-1}(T_t-T_e) - n T_t$, which is negative for $\tau\in[0,1]$. Thus, $T_s - n \tau T_t$ equals 0 at $\tau = \tau_{opt}$ and decreases afterwards, which implies that $T_s - n \tau T_t<0$ for $\tau > \tau_{opt}$. With this, (\ref{eq-derivative}) can be rewritten as
\begin{equation}
r_{opt} - r(\tau_M) = - \int_{\tau_M}^{\tau_{opt}}{\bigg| \frac{d r(\tau)}{d \tau}\bigg| d\tau}
\end{equation}
which can be bounded as follows:
\small
\begin{align}
r_{opt} - r(\tau_M) & \leq - \int_{\tau_{M}}^{\tau_{opt}}{\bigg| \frac{dr(\tau)}{d \tau}\bigg|_{\max} d\tau}  = \bigg| \frac{dr(\tau)}{d \tau}\bigg|_{\max}(\tau_M - \tau_{opt})
\label{eq-bound-derivative}
\end{align}
\normalsize
where $|dr(\tau)/d\tau|_{\max}$ is an upper bound for the absolute value the derivative in the interval $\tau \in [\tau_{opt}, \tau_{M}]$


To find $|dr(\tau)/d\tau|_{\max}$, we proceed as follows. Given that $\tau \in \left(\tau_{opt}, \tau_M \right]$ and $\tau_M \leq \min(2\tau_{opt},1)$, we want to evaluate $dr(\tau)/d\tau$ at $\tau = K \tau_{opt}$ for $1 < K \leq \min(2,1/\tau_{opt})$, which yields
\begin{equation}
\frac{\partial r(\tau)}{\partial \tau} = \frac{l(1-K \tau_{opt})^{n-2}(T_K - n K \tau_{opt} T_t)}{T_K^2}
\end{equation}
where $T_K$ is the value of $T_s$ for $\tau_i = K \tau_{opt} \ \forall i$. Note that, for $K > 1$, we have $T_K > T_{opt}$ and $T_K - n K \tau_{opt} T_t < 0$ (the latter holds since we have earlier shown that the term $T_s - n \tau T_t$ is negative for $\tau > \tau_{opt}$). With this, the absolute value of $dr(\tau)/d\tau$ can be bounded by
\begin{equation}
\bigg| \frac{dr(\tau)}{d \tau}\bigg| \leq \frac{l(1 -\tau_{opt})^{n-2}(n K \tau_{opt} T_t - T_{opt})}{T_{opt}^2}
\end{equation}

Before, we have shown that the term $T_s - n \tau T_t$ is equal to 0 at $\tau = \tau_{opt}$, i.e., $T_{opt} - n\tau_{opt}T_t = 0$. Adding this term to $n K \tau_{opt} T_t - T_{opt}$ gives $(K-1)n \tau_{opt} T_t$. Furthermore, since $T_{opt} = n\tau_{opt}T_t$, this can be expressed as $(K-1)T_{opt}$. Combining this with the above equation yields:
\begin{align}
\bigg| \frac{dr(\tau)}{d \tau}\bigg| \leq & \ \frac{l(1-\tau_{opt})^{n-2}(K-1)}{T_{opt}} \leq \frac{l(1-\tau_{opt})^{n-2}}{T_{opt}} 
\end{align}
Finally, combining the above bound on the maximum value of the derivative with (\ref{eq-bound-derivative}) leads to:
\small
\begin{align*}
r_{opt} - \frac{1}{n}\sum_j r_j 
&  \leq \frac{l(1-\tau_{opt})^{n-2}}{T_{opt}} (\tau_M - \tau_{opt}) 
 = \frac{r_{opt}}{\tau_{opt}(1-\tau_{opt})}\Delta
\end{align*}
\normalsize
from which (\ref{eq-F_bound_3}) also holds for this case.
\end{proof}


\subsection{Proof of Theorem \ref{th:one}}\label{B}

Once again, we proceed by establishing a number of intermediate Lemmas, and then present the proof of Theorem  \ref{th:one}.

\begin{lemma}\label{lem:one}
\begin{align*}
(i)\ &\sum_{j \neq i} \left(r_j(\boldsymbol{\hat{\tau}}) - r_i(\boldsymbol{\hat{\tau}})\right)
\le 
\frac{\left(n - 1\right)l}{T_m}\left(\hat{\tau}_M - \hat{\tau}_i\right) \left(1-\frac{\tau_{opt}}{2}\right)^{n-2}\\
(ii)\ &\sum_{j \neq i} \left(r_j(\boldsymbol{\hat{\tau}}) - r_i(\boldsymbol{\hat{\tau}})\right)
\ge 
\frac{\left(n - 1\right)l}{T_m}\left(\hat{\tau}_m - \hat{\tau}_i\right) \left(1-\frac{\tau_{opt}}{2}\right)^{n-2}
\end{align*}
with $n\ge2$, $\hat{\tau}_M = \max_{i\in\{1,\cdots,n\}} \hat{\tau}_i$, $\hat{\tau}_m = \min_{i\in\{1,\cdots,n\}} \hat{\tau}_i$, $\frac{\tau_{opt}}{2} \le \hat{\tau}_k \le 1$.
\end{lemma}
\begin{proof}
($i$) Since $r_i\le r_M$ we have $\sum_{j\neq i} \left(r_j - r_i\right) \le \sum_{j \neq i} \left(r_M- r_i\right) = (n-1)(r_M-r_i)$.  Substituting from (\ref{eq-r}) and rearranging we have
\begin{align*}
(n-1)(r_M-r_i)
=
\frac{\left(n - 1\right)l}{T_s}\left(\hat{\tau}_M - \hat{\tau}_i\right) \prod_{k \neq i, M} \! \left(1-\hat{\tau}_k\right)\\
\stackrel{(a)}{\le} \frac{\left(n - 1\right)l}{T_m}\left(\hat{\tau}_M - \hat{\tau}_i\right) \left(1-\frac{\tau_{opt}}{2}\right)^{n-2}
\end{align*}
where $(a)$ follows from the fact that $\frac{\tau_{opt}}{2} \le \hat{\tau}_k \le 1$ and $T_s\ge T_m:= T_t + (T_e - T_t)\left(1-\frac{\tau_{opt}}{2}\right)^n$ (the latter holds since $\hat{\tau}_j\ge \frac{\tau_{opt}}{2}$).

($ii$) Since $r_i\ge r_m$ we have $\sum_{j\neq i} \left(r_j - r_i\right) \ge  (n-1)(r_m-r_i)$.  The second part of the result now follows using an identical argument to ($i$).
\end{proof}

\begin{lemma}\label{lem:six}
(i) $\hat{\tau}_m - \hat{\tau}_i \ge \tau_m - \tau_i$ and (ii) $\hat{\tau}_M - \hat{\tau}_i \le \tau_M - \tau_i $, where  $\tau_m = \min_{i\in\{1,\cdots,n\}} \tau_i$, $\tau_M = \max_{i\in\{1,\cdots,n\}} \tau_i$, $\hat{\tau}_j = \max\{\frac{\tau_{opt}}{2}, \tau_j\}$.
\end{lemma}
\begin{proof}
($i$) When $\tau_i \ge \frac{\tau_{opt}}{2}$ then $\hat{\tau}_i=\tau_i$.  Since $\hat{\tau}_m\ge \tau_m$ it follows that $\hat{\tau}_m - \hat{\tau}_i \ge \tau_m - \tau_i$.   When $\tau_i < \frac{\tau_{opt}}{2}$, then $\hat{\tau}_m = \hat{\tau}_i = \frac{\tau_{opt}}{2}$, and hence $\hat{\tau}_m - \hat{\tau}_i = 0$, while $\tau_m - \tau_i\le 0$.   
($ii$) When $\tau_i \ge \frac{\tau_{opt}}{2}$ then $\hat{\tau}_i=\tau_i$, $\hat{\tau}_M=\tau_M$ and $\hat{\tau}_M - \hat{\tau}_i = \tau_M - \tau_i $.  When $\tau_i < \frac{\tau_{opt}}{2}$ we have two cases:  (a) if $\tau_M<\frac{\tau_{opt}}{2}$ then $\hat{\tau}_M - \hat{\tau}_i = 0 \le \tau_M - \tau_i $;  (b) if $\tau_M\ge\frac{\tau_{opt}}{2}$ then $\hat{\tau}_M - \hat{\tau}_i = \tau_M-\hat{\tau}_i\le \tau_M-\tau_i$ since $\tau_i \le \hat{\tau}_i$.
\end{proof}

\begin{lemma}\label{lem:seven}
When $\gamma < \gamma_{max}= \left(\frac{nl}{T_m}\left(1-\frac{\tau_{opt}}{2}\right)^{n-2}\right)^{-1} $, $D(\boldsymbol{\hat{\tau}})<0$ and $n\ge 2$ then under update (\ref{eq-tau_updt}),
\begin{equation}
\tau_i\left(t+1\right) < \tau_M\left(t+1\right) \text{ if } \tau_i\left(t\right) < \tau_M\left(t\right)
\end{equation}
where $\tau_M = \max_{i\in\{1,\cdots,n\}} \tau_i$
\end{lemma}
\begin{proof}
It is sufficient to show that
\small
\begin{align*}
\tau_i + \gamma \left( \sum_{j \neq i} \left(r_j - r_i\right) - F_i \right) < 
 \tau_M + \gamma \left( \sum_{j' \neq M} \left(r_j' - r_M\right) - F_M \right)
\end{align*}
\normalsize
where we drop the $t$ arguments from all quantities to streamline notation.  Since $F_i = F_M$, $i=1,\cdots,n$ when $D<0$ this simplifies to $\gamma n \left( r_M - r_i\right) < \tau_M - \tau_i$.  Substituting from (\ref{eq-r}) we obtain
\begin{equation}
\gamma \frac{nl}{T_s} \left(\hat{\tau}_M - \hat{\tau}_i\right) \prod_{j \neq i,M} \left( 1 - \hat{\tau}_j\right) < \tau_M  - \tau_i
\label{eq-cond_tauM_F4*}
\end{equation}
By Lemma \ref{lem:six}, $\hat{\tau}_M - \hat{\tau}_i \leq \tau_M - \tau_i$ and a sufficient condition for (\ref{eq-cond_tauM_F4*}) is
$\gamma \frac{nl}{T_s}\prod_{j \neq i,M} \left( 1 - \hat{\tau}_j\right) < 1$.  Since $\hat{\tau}_j \ge \frac{\tau_{opt}}{2}$, this  holds when $\gamma < \gamma_{max}$.
\end{proof}

\begin{lemma}\label{lem:eight}
Under the conditions of Lemma \ref{lem:seven}, $\tau_M \left(t+1\right) \le \tau_M\left(t\right)$, with equality only when $\tau_j=\tau_{opt}$, $j=1,\cdots,n$.
\end{lemma}
\begin{proof}
It is sufficient to show that
\begin{equation}
\tau_M + \gamma \left( \displaystyle \sum_{i} r_i - n r_M - \frac{nr_{opt} - \sum_i r_i}{n-1} \right) \le \tau_M
\end{equation}
with equality only when $\tau_j=\tau_{opt}$, $j=1,\cdots,n$.    When $\tau_j=\tau_{opt}$, equality holds.  Assume now that  $\tau_j\ne\tau_{opt}$ for some $j$.   Since $\gamma > 0$, the above condition is satisfied when
\begin{equation}
r_M + \displaystyle \sum_{i \neq M} \left( r_i - r_M \right) - r_{opt} < 0
\label{eq-condition}
\end{equation}
If $\hat{\tau}_M = 1$, then (\ref{eq-condition}) is satisfied since $n>1$ and $r_i = 0$ for $i \neq M$.  Suppose therefore $\hat{\tau}_M < 1$ and define function $G = r_M + \sum_{i \neq M} \left( r_i - r_M \right) - r_{opt}$.  The partial derivative of $G$ with respect to $\hat{\tau}_i$ is given by 
$
\frac{\partial G}{\partial\hat{\tau}_i} = \frac{\partial r_i}{\partial \hat{\tau}_i} + \sum_{j \neq i, M}{\frac{\partial (r_j - r_M)}{\partial \hat{\tau}_i}}
$.
It can be verified that $\partial r_i/\partial \hat{\tau}_i > 0$ (since $\hat{\tau}_j\le\hat{\tau}_M < 1$). Also,
\begin{align*}
\frac{\partial (r_j - r_M)}{\partial \hat{\tau}_i} = & - \frac{l}{T_s^2}\left(\frac{\hat{\tau}_j}{1- \hat{\tau}_j} - \frac{\hat{\tau}_M}{1- \hat{\tau}_M}\right) \nonumber\\
&\times\bigg(\prod_{k \neq i}{(1 - \hat{\tau}_k)} T_s 
+ \prod_{k}{(1 - \hat{\tau}_k)\frac{\partial T_s}{\partial \hat{\tau}_i}}\bigg)
\ge 0
\end{align*}
since $\partial T_s/\partial \hat{\tau}_i > 0$ and $\frac{\tau}{1-\tau}$ is monotonically increasing in $\tau$.  Hence
 $\partial G/\partial \hat{\tau}_i > 0$, which implies that $G$ takes a maximum for the largest possible value of $\hat{\tau}_i$, for all $i \neq M$. Since $\hat{\tau}_i \leq \hat{\tau}_M$, this means that $G$ is maximized when $\hat{\tau}_i = \hat{\tau}_M$ for all $i$. In this case, (\ref{eq-condition}) becomes $r_M - r_{opt} < 0$.   Since $r_{opt}$ is the maximum throughput when all stations use the same transmission attempt probability, $r_M - r_{opt} = 0$ only if $\hat{\tau}_M=\tau_{opt}$. But by assumption $\tau_M \ne \tau_{opt}$ and so we must have $r_M-r_{opt}<0$.  
\end{proof}

\begin{proof}[Proof of Theorem \ref{th:one}]
To establish global asymptotically stability we show that $\|\boldsymbol\tau\left(t+1\right) - \boldsymbol\tau_{opt}\|_\infty < \|\boldsymbol\tau\left(t\right) - \boldsymbol\tau_{opt}\|_\infty$ unless $\boldsymbol\tau\left(t\right) = \boldsymbol\tau_{opt}$. By definition, $\|\boldsymbol\tau \left(t\right) - \boldsymbol\tau_{opt}\|_\infty=\max\left(\left| \tau_M \left(t\right) - \tau_{opt}\right|, \left| \tau_m\left(t\right) - \tau_{opt}\right|\right)$, where $\tau_M$ and $\tau_m$ are the maximum and minimum values of the elements of vector $\boldsymbol\tau$ respectively.   We proceed in a case-by-case fashion.

\textbf{Case 1}: $\tau_M\left(t\right) > \tau_{opt}$, $\|\boldsymbol\tau \left(t\right) - \boldsymbol\tau_{opt}\|_\infty = \tau_M\left(t\right)-\tau_{opt}$. For $\|\boldsymbol\tau\left(t+1\right) - \boldsymbol\tau_{opt}\|_\infty < \|\boldsymbol\tau\left(t\right) - \boldsymbol\tau_{opt}\|_\infty$ we require
\begin{equation}
\left|\tau_i \left(t+1\right) - \tau_{opt}\right| < \tau_M\left(t\right) - \tau_{opt}, \qquad  i=1,\cdots,n
\label{eq-cond_tauM}
\end{equation}
Substituting from (\ref{eq-tau_updt}) and (\ref{eq-ei}), (\ref{eq-cond_tauM}) is satisfied provided:
\begin{align}
\tau_i + \gamma \left( \displaystyle \sum_{j \neq i} \left(r_j - r_i\right) - F_i \right) - \tau_{opt} &< \tau_M - \tau_{opt}
\label{eq-C1}\\
\tau_i + \gamma \left( \displaystyle \sum_{j \neq i} \left(r_j - r_i\right) - F_i \right) - \tau_{opt} &> \tau_{opt} - \tau_M
\label{eq-C2}
\end{align}
where the dependency on $t$ has been omitted to simplify notation. 

\textbf{Case 1a ($F_i = -D/n$: $\tau_i \leq \tau_{opt}$, $nr_{opt} \geq \sum_j r_j$)}.
Using Lemma \ref{lem:one} plus Theorem \ref{lem:five} with $\Delta=\tau_M-\tau_{opt}$, (\ref{eq-C1}) is satisfied provided
\footnotesize
\begin{align}
&\gamma  \left(\frac{\left(n - 1\right)l}{T_m}\left(\hat{\tau}_M - \hat{\tau}_i\right)  \left(1-\frac{\tau_{opt}}{2}\right)^{n-2} 
+\rho(\tau_M-\tau_{opt})\right)  < \tau_M - \tau_i\label{eq:DL0}
\end{align}
\normalsize
By Lemma \ref{lem:six}, $\hat{\tau}_M -\hat{\tau}_i \leq \tau_M - \tau_i$. Also, by assumption $\tau_i\le\tau_{opt}$ and so $\tau_M-\tau_i \geq \tau_M-\tau_{opt}$. It then follows that (\ref{eq:DL0}) (and so (\ref{eq-C1})) is satisfied provided
\begin{equation}
\gamma < \gamma_{max} < \left(\frac{\left(n - 1\right)l}{T_m}\left(1-\frac{\tau_{opt}}{2}\right)^{n-2} + \rho\right)^{-1}
\label{eq-gamma0}
\end{equation}

Since $-F_i =\frac{D}{n}\ge 0$, (\ref{eq-C2}) is satisfied provided
\begin{align}\label{eq:DL1}
\gamma \displaystyle \sum_{j \neq i} \left(r_j - r_i\right)  > 2\tau_{opt} - \tau_M - \tau_i
\end{align}
By assumption $\left|\tau_M-\tau_{opt}\right| \geq \left|\tau_m-\tau_{opt}\right|$ and so $\tau_m \geq 2\tau_{opt} - \tau_M$.  Also, by assumption $\tau_i\le \tau_{opt} <\tau_M$ and so $\tau_m < \tau_M$, $r_m<r_M$.  If $\tau_i =2\tau_{opt}-\tau_M$ then $r_i=r_m$ and (\ref{eq:DL1}) holds provided $\gamma>0$ (the RHS equals 0 while the LHS is lower bounded by $\gamma(r_M-r_m)>0$).  Otherwise, suppose $\tau_i > 2\tau_{opt}-\tau_M$.  By Lemma \ref{lem:one}, (\ref{eq:DL1}) is satisfied provided 
\begin{align*}
\gamma\frac{\left(n - 1\right)l}{T_m}\left(\hat{\tau}_m - \hat{\tau}_i\right) \left(1-\frac{\tau_{opt}}{2}\right)^{n-2}
>
2\tau_{opt} - \tau_M - \tau_i
\end{align*}
By Lemma \ref{lem:six}, $\hat{\tau}_m - \hat{\tau}_i \ge \tau_m - \tau_i$.   And  $\tau_m-\tau_i \ge 2\tau_{opt} - \tau_M - \tau_i$.  Hence, (\ref{eq:DL1}) is satisfied provided 
\begin{equation}
\gamma  < \frac{n}{n-1}\gamma_{max}=\left(\frac{\left(n - 1\right)l}{T_m}\left(1-\frac{\tau_{opt}}{2}\right)^{n-2} \right)^{-1}
\label{eq-gamma2}
\end{equation}


\textbf{Case 1b ($F_i = \frac{D}{n}$: $\tau_i > \tau_{opt}$, $nr_{opt} \geq \sum_j r_j$)}.
From $nr_{opt} \geq \sum_j r_j$, it holds that $F_i = \frac{D}{n}=r_{opt} - \frac{1}{n}\sum_{j}r_j \geq 0$. If $\tau_i = \tau_M$ then either ($i$) $\tau_j = \tau_i$ for all $j$, in which case $F_i>0$ and $\sum_{j \neq i} \left(r_j - r_i\right) = 0$, or ($ii$) $\tau_j < \tau_i$ for some $j$, in which case $\sum_{j \neq i} \left(r_j - r_i\right) < 0$ and (as mentioned above) $F_i \geq 0$. In both cases, (\ref{eq-C1}) is satisfied. 
Otherwise, we have $\tau_i < \tau_M$.  Since $F_i \ge 0$, (\ref{eq-C1}) is satisfied provided $\gamma  \displaystyle \sum_{j \neq i} \left(r_j - r_i\right)  < \tau_M -\tau_i$.  By Lemma \ref{lem:one}, this holds provided
\begin{align}\label{eq:DL2}
\gamma \frac{\left(n - 1\right)l}{T_m}\left(\hat{\tau}_M - \hat{\tau}_i\right) \left(1-\frac{\tau_{opt}}{2}\right)^{n-2}
<
\tau_M -\tau_i
\end{align}
By assumption, $\tau_i > \tau_{opt}$ and so $\hat{\tau}_i=\tau_i$, $\hat{\tau}_M=\tau_M$.   Also, $\tau_M -\tau_i>0$.  Hence, (\ref{eq:DL2}) holds when $\gamma$ satisfies (\ref{eq-gamma2}).

Using Lemma \ref{lem:one} plus Theorem \ref{lem:five} with $\Delta=\tau_M-\tau_{opt}$, (\ref{eq-C2}) is satisfied provided
\footnotesize
\begin{align}
&\gamma \left( \frac{\left(n - 1\right)l}{T_m}\left(\hat{\tau}_i - \hat{\tau}_m\right)  \left(1-\frac{\tau_{opt}}{2}\right)^{n-2} 
+ \rho(\tau_M-\tau_{opt})\right) \nonumber\\
&\qquad < \tau_i - (2\tau_{opt} - \tau_M)
\label{eq-cond_taum_F3'}
\end{align}
\normalsize
By assumption, $\tau_m \geq 2\tau_{opt} - \tau_M$ and so by Lemma \ref{lem:six}, $\hat{\tau}_i - \hat{\tau}_m \le \tau_i - (2\tau_{opt} - \tau_M)$.   Also, by assumption $\tau_i > \tau_{opt}$ and so $\tau_i - (2\tau_{opt} - \tau_M) \ge \tau_M- \tau_{opt}$.  It then follows that 
(\ref{eq-cond_taum_F3'}) (and so (\ref{eq-C2})) is satisfied when $\gamma$ satisfies (\ref{eq-gamma0}).


\textbf{Case 1c ($F_i = D/\left(n-1\right)$: $nr_{opt} < \sum_j r_j$)}.
By Lemmas \ref{lem:seven} and \ref{lem:eight}, $\tau_i\left(t+1\right) <\tau_M \left(t+1\right) < \tau_M\left(t\right)$ and so (\ref{eq-C1}) is satisfied (observe that the LHS of (\ref{eq-C1}) is $\tau_i\left(t+1\right)-\tau_{opt}$).  Since $F_i\le 0$, (\ref{eq-C2}) is satisfied provided $\gamma \sum_{j \neq i} \left(r_j - r_i\right)   > 2\tau_{opt} - \tau_M -\tau_i$.  By Lemma \ref{lem:one}, this holds provided 
\begin{align*}
\gamma \frac{\left(n - 1\right)l}{T_m}\left(\hat{\tau}_i - \hat{\tau}_m\right) \left(1-\frac{\tau_{opt}}{2}\right)^{n-2} < \tau_i -( 2\tau_{opt} - \tau_M)
\end{align*}
By assumption, $\tau_m \geq 2\tau_{opt} - \tau_M$ and by Lemma \ref{lem:six}, $\hat{\tau}_i - \hat{\tau}_m \le \tau_i - (2\tau_{opt} - \tau_M)$.  Hence, the above holds (and so (\ref{eq-C2}) is satisfied) when $\gamma$ satisfies (\ref{eq-gamma2}). 


\textbf{Case 2}: $\tau_M\left(t\right) < \tau_{opt}$ or $\|\boldsymbol\tau \left(t\right) - \boldsymbol\tau_{opt}\|_\infty \neq \tau_M\left(t\right)-\tau_{opt}$.  In this case, it necessarily holds that $\tau_m \left(t\right) < \tau_{opt}$ and $\|\boldsymbol\tau \left(t\right) - \boldsymbol\tau_{opt}\|_\infty = \tau_{opt} - \tau_m\left(t\right)$.    For $\|\boldsymbol\tau\left(t+1\right) - \boldsymbol\tau_{opt}\|_\infty < \|\boldsymbol\tau\left(t\right) - \boldsymbol\tau_{opt}\|_\infty$ we require
\begin{align}
\tau_i + \gamma \left( \displaystyle \sum_{j \neq i} \left(r_j - r_i\right) - F_i \right) - \tau_{opt} &< \tau_{opt} -\tau_m
\label{eq-D1}\\
\tau_i + \gamma \left( \displaystyle \sum_{j \neq i} \left(r_j - r_i\right) - F_i \right) - \tau_{opt} &> \tau_m -\tau_{opt}
\label{eq-D2}
\end{align}


\textbf{Case 2a ($F_i = -D/n$: $\tau_i \leq \tau_{opt}$, $nr_{opt} \geq \sum_j r_j$)}.
By Lemma \ref{lem:one} and Theorem \ref{lem:five} with $\Delta=\tau_{opt}-\tau_m$ condition (\ref{eq-D1}) is satisfied provided
\footnotesize
\begin{align}\label{eq:DL21}
&\gamma  \left(\frac{\left(n - 1\right)l}{T_m}\left(\hat{\tau}_M - \hat{\tau}_i\right)  \left(1-\frac{\tau_{opt}}{2}\right)^{n-2} 
+ \rho(\tau_{opt}-\tau_{m})\right) \nonumber \\
&\qquad<
2\tau_{opt} - \tau_m - \tau_i 
\end{align}
\normalsize
Since $\tau_i \leq \tau_{opt}$ then $2\tau_{opt} - \tau_m - \tau_i \ge \tau_{opt}-\tau_m$.  By Lemma \ref{lem:six}, $\hat{\tau}_M - \hat{\tau}_i \le \tau_M - \tau_i $.   When $\|\boldsymbol\tau \left(t\right) - \boldsymbol\tau_{opt}\|_\infty \neq \tau_M\left(t\right)-\tau_{opt}$, then $\left|\tau_M-\tau_{opt}\right| < \left|\tau_m-\tau_{opt}\right|$ and so $\tau_m < 2\tau_{opt} - \tau_M$ i.e. $\tau_M < 2\tau_{opt} - \tau_m$.  Hence,  $\tau_M - \tau_i < 2\tau_{opt} -\tau_m -\tau_i$.   When $\tau_M < \tau_{opt}$ then $\tau_M - \tau_i \le \tau_{opt} - \tau_i = 2\tau_{opt} -\tau_m -\tau_i -(\tau_{opt}-\tau_m)\le 2\tau_{opt} -\tau_m -\tau_i$ where the last inequality follows from the fact that $\tau_m\le \tau_i\le \tau_{opt}$.   It follows that (\ref{eq:DL21}) holds when $\gamma$ satisfies  (\ref{eq-gamma0}).


If $\tau_i=\tau_m$ then  (\ref{eq-D2}) is satisfied since $r_j-r_m\ge 0$ and $F_i>0$ (unless $\tau_i=\tau_{opt}$ $\forall i$). Otherwise, if $\tau_i>\tau_m$ then since $F_i\le 0$, (\ref{eq-D2}) is satisfied provided $\gamma \sum_{j \neq i} \left(r_j - r_i\right) >\tau_m-\tau_i$.     By Lemmas \ref{lem:one} and \ref{lem:six}, this holds when $\gamma$ satisfies (\ref{eq-gamma2}).


\textbf{Case 2b ($F_i = \frac{D}{n}$: $\tau_i > \tau_{opt}$, $nr_{opt} \geq \sum_j r_j$)}.
Note that $\tau_M \ge \tau_i >\tau_{opt}$.  Hence, to be in case 2 we must have $\left|\tau_M-\tau_{opt}\right| < \left|\tau_m-\tau_{opt}\right|$ and so $\tau_m < 2\tau_{opt} - \tau_M$ i.e. $\tau_M < 2\tau_{opt} - \tau_m$. If $\tau_i=\tau_M$ then (\ref{eq-D1}) is satisfied since the LHS non-negative while the RHS is positive.  Otherwise, if $\tau_i<\tau_M$ then since $F_i\ge 0$, (\ref{eq-D1}) is satisfied provided $\gamma \sum_{j \neq i} \left(r_j - r_i\right) < 2\tau_{opt} -\tau_m -\tau_i$.  By Lemma \ref{lem:one}  this holds when
\begin{align}\label{eq:DL_20}
\frac{\left(n - 1\right)l}{T_m}\left(\hat{\tau}_M - \hat{\tau}_i\right) \left(1-\frac{\tau_{opt}}{2}\right)^{n-2}
< 
2\tau_{opt} -\tau_m -\tau_i
\end{align}
As already noted, $\tau_M < 2\tau_{opt} - \tau_m$ and $\tau_M\ge \tau_i > \tau_{opt}$.  Hence,  $\hat{\tau}_M - \hat{\tau}_i = \tau_M - \tau_i < 2\tau_{opt} -\tau_m -\tau_i$.  It follows that  (\ref{eq:DL_20}) is satisfied when $\gamma$ satisfies (\ref{eq-gamma2}).   


Using Lemma \ref{lem:one} and Theorem \ref{lem:five} with $\Delta=\tau_{opt}-\tau_m$, condition (\ref{eq-D2}) is satisfied provided
\small
\begin{align*}
&\gamma  \frac{\left(n - 1\right)l}{T_m}\left(\hat{\tau}_m - \hat{\tau}_i\right)  \left(1-\frac{\tau_{opt}}{2}\right)^{n-2} 
 +\gamma\rho(\tau_m-\tau_{opt})  > \tau_m - \tau_i
\end{align*}
\normalsize
Since $\tau_i>\tau_{opt}$, $\hat{\tau}_m - \hat{\tau}_i\ge \tau_m-\tau_i$.  Also, $\tau_m-\tau_{opt}\ge \tau_m-\tau_i$.  It follows that the above holds when $\gamma$ satisfies (\ref{eq-gamma0}).


\textbf{Case 2c ($F_i = D/\left(n-1\right)$: $nr_{opt} < \sum_j r_j$)}.
Observe that the LHS of (\ref{eq-D1}) is $\tau_i\left(t+1\right)-\tau_{opt}$. By Lemmas \ref{lem:seven} and \ref{lem:eight}, $\tau_i\left(t+1\right) <\tau_M \left(t+1\right) < \tau_M\left(t\right)$. By assumption, $\tau_M\left(t\right) - \tau_{opt} < \tau_{opt} - \tau_m$. Therefore, (\ref{eq-D1}) is satisfied.

If $\tau_i=\tau_m$ then (\ref{eq-D2}) is satisfied.  Otherwise, if $\tau_i>\tau_m$ then since $F_i\le 0$ condition (\ref{eq-D2}) is satisfied provided $\gamma \sum_{j \neq i} \left(r_j - r_i\right) > \tau_m -\tau_i$.  By  Lemma \ref{lem:one} this holds provided $\gamma \frac{\left(n - 1\right)l}{T_m}\left(\hat{\tau}_i - \hat{\tau}_m\right) \left(1-\frac{\tau_{opt}}{2}\right)^{n-2} < \tau_i -\tau_m$.  By Lemma \ref{lem:six} this holds when $\gamma$ satisfies (\ref{eq-gamma2}). 
\end{proof}
\color{black}

\subsection{Proofs of Theorems \ref{th:two} and \ref{th:three} }
\begin{proof}[Proof of Theorem \ref{th:two}]
The GAS algorithm computes $\tau_i$ at a given stage $t'$ according to the following expression:
\begin{equation}
\tau_i(t') = \tau_i^{initial} + \gamma \sum_{t = 0}^{t'}{\left(\sum_{j \neq i}{\left(r_j(t)-r_i(t)\right)-F_i(t)}\right)}
\label{eq-tauit'}
\end{equation}

If $\tau_i$ exceeds 1 at any stage, then it decreases in the next stages until it goes below 1. Indeed, for $\tau_i > 1$ we have $\hat{\tau}_i = 1$, which leads to $r_j = 0$ for $j\neq i$ and $F_i > -r_i$, and thus from the above expression $\tau_i$ decreases. This implies that $\tau_i$ can never exceed $1+\delta$, where $\delta$ is the maximum distance that $\tau_i$ can cover in one stage. From (\ref{eq-ei}), we have that $\delta \leq \gamma\max(C+r_{opt},n(C-r_{opt})/(n-1))$, where $C$ is the maximum total throughput of the WLAN. Therefore, $\tau_i$ never exceeds $\tau_{max} = 1 + \gamma\max(C+r_{opt},n(C-r_{opt})/(n-1))$. Taking this into account, (\ref{eq-tauit'}) yields
\begin{equation}\label{eq-ineq1}
\sum_{t}{\left(\sum_{j \neq i}{\left(r_j(t)-r_i(t)\right)-F_i(t)}\right)} \leq \frac{\tau_{max}-\tau_i^{initial}}{\gamma}
\end{equation}

Let us consider the case in which there is a selfish station that changes its configuration over time and receives a throughput $r_s(t)$ while the rest of the stations are well-behaved, using the same configuration and obtaining the same throughput $r(t)$. Then the above can be expressed as
\begin{equation}\label{eq-ineq}
\sum_{t}{r_s(t)} \leq \sum_{t}{\left(r(t) + F_i(t)\right)} + \frac{\tau_{max}-\tau_i^{initial}}{\gamma}
\end{equation}

If we now consider the throughput of the selfish user over an interval $T$, the average throughput over this interval can be computed as:
\begin{equation}
r_s = \frac{1}{T} \sum_{t}{r_s(t) T_{beacon}}
\end{equation}

From (\ref{eq-ineq}),
\begin{equation}
r_s \leq \frac{1}{T} \sum_{t}{\left(r(t) + F_i(t)\right)T_{beacon}} + \left(\frac{\tau_{max}-\tau_i^{initial}}{\gamma}\right)\frac{T_{beacon}}{T}
\end{equation}

Since we considering a very large interval $T \to \infty$, the term $\left(\frac{\tau_{max}-\tau_i^{initial}}{\gamma}\right)\frac{T_{beacon}}{T}$ tends to 0, which yields
\begin{equation}\label{eq-ss2}
r_s \leq \frac{1}{T} \sum_{t}{\left(r(t) + F_i(t)\right)T_{beacon}}
\end{equation}

Let us consider now a given stage $t$. From (\ref{eq-F}) we have
\begin{equation}
F_i(t) \leq \frac{1}{n-1}\left(nr_{opt}-r_s(t)-(n-1)r(t)\right)
\end{equation}
which yields
\begin{equation}
(n-1)r(t) + r_s(t) + (n-1)F_i(t) \leq nr_{opt}
\end{equation}

Since the above equation is satisfied for all $t$,
\begin{equation}
\sum_t{(n-1)r(t) + r_s(t) + (n-1)F_i(t)} \leq \sum_t{nr_{opt}}
\end{equation}

Furthermore, from (\ref{eq-ss2}),
\begin{equation}
(n-1)\sum_{t}{r_s(t)} \leq (n-1)\sum_{t}{(r(t) + F_i(t))}
\end{equation}

Adding the above two equations yields
\begin{equation}
n \sum_{t}{r_s(t)} \leq n\sum_{t}{r_{opt}}
\end{equation}
from which
\begin{equation}
r_s = \frac{1}{T} \sum_{t}{r_s(t)T_{beacon}} \leq \frac{1}{T} \sum_{t}{r_{opt}T_{beacon}} = r_{opt}
\end{equation}
which proves the theorem. Since the right hand side of the above equation is precisely the throughput that the selfish station would get if it always played GAS, this shows that the selfish station cannot benefit from using a different strategy no matter how it changes its configuration over time. As the proof does not make any assumption on the configuration of the selfish station, this holds for any configuration of all the 802.11e parameters.
\end{proof}


\begin{proof}[Proof of Theorem \ref{th:three}]
The proof of Theorem \ref{th:two} is independent of the past history, and therefore it can be applied to any subgame. This means that \emph{All-GAS} is a Nash equilibrium of any subgame.
\end{proof}

\vfill

\end{document}